\documentclass[12pt,reqno]{amsart}

\usepackage{bm}
\usepackage{mathrsfs}
\usepackage{amsmath}
\usepackage{amssymb}
\usepackage{amsfonts}
\usepackage{amsthm}
\usepackage{a4wide}
\usepackage{graphicx}
\usepackage{color}
\usepackage[sans]{dsfont}
\usepackage{linearA}
\usepackage{longtable}
\usepackage{pdflscape}
\usepackage{float}
\usepackage{cancel}
\usepackage{dsfont}
\usepackage{setspace}
\usepackage{array}
\usepackage{hyperref}
\hypersetup{colorlinks=true,allcolors=[rgb]{0,0,0.6}}
\usepackage[english]{babel}
\usepackage[applemac]{inputenc}
\usepackage[T1]{fontenc}
\usepackage{tikz}
\usepackage{enumitem}
\usepackage[square,numbers,sort&compress]{natbib}
\usepackage[left=2.5cm,right=2.5cm,top=3cm,bottom=3cm]{geometry}
\usepackage{soul}
\usepackage[sort&compress,capitalise,nameinlink]{cleveref}
\crefname{section}{\textsection}{\textsection}
\crefname{subsection}{\textsection}{\textsection}
\crefname{subsubsection}{\textsection}{\textsection}
\crefname{paragraph}{\textparagraph}{\textparagraph}
\crefname{thm}{Theorem}{Theorems}
\crefname{assumption}{Assumption}{Assumptions}
\crefname{prop}{Proposition}{Propositions}

\newcommand{\one}{\mathds{1}}
\newcommand{\N}{\mathbb{N}}
\newcommand{\C}{\mathbb{C}}
\newcommand{\R}{\mathbb{R}}

\newcommand{\vecp}{\vec{P}}
\def\D{\mathrm{d}}
\newcommand{\Hnele}{{H^{\rm{Nel}}_\varepsilon}}
\newcommand{\Hnelm}{{H^{\rm{Nel}}_\mu}}

\newcommand{\Hpfm}{{H^{\rm{PF}}_\mu}}

\renewcommand{\vec}{\mathbf}
\renewcommand{\le}{\leqslant}

\renewcommand{\leq}{\leqslant}
\renewcommand{\geq}{\geqslant}

\newcommand{\lf}{\left}
\newcommand{\ri}{\right}
\newcommand{\disp}{\displaystyle}

\newcommand{\braket}[2]{\lf\langle #1|#2 \ri\rangle}
\newcommand{\braketr}[2]{\lf\langle #1\lf|#2\ri. \ri\rangle}

\newcommand{\meanlr}[3]{\lf\langle #1\lf|#2\ri|#3\ri\rangle}
\newcommand{\eps}{\varepsilon}

\newcommand{\diff}{\mathrm{d}}
\newcommand{\sigmav}{\bm{\sigma}}
\numberwithin{equation}{section}
\newcommand{\bdm}{\begin{displaymath}}
\newcommand{\edm}{\end{displaymath}}
\newcommand{\bay}{\begin{array}{c}}
\newcommand{\eay}{\end{array}}
\newcommand{\ben}{\begin{enumerate}}
\newcommand{\een}{\end{enumerate}}
\newcommand{\beq}{\begin{equation}}
\newcommand{\eeq}{\end{equation}}
\newcommand{\beqn}{\begin{eqnarray}}
\newcommand{\eeqn}{\end{eqnarray}}
\newcommand{\bml}[1]{\begin{multline} #1 \end{multline}}
\newcommand{\bmln}[1]{\begin{multline*} #1 \end{multline*}}

\renewcommand{\Re}{\mathrm{Re}}

\newcommand{\lyxmathsym}[1]{\ifmmode\begingroup\def\b@ld{bold}
  \text{\ifx\math@version\b@ld\bfseries\fi#1}\endgroup\else#1\fi}

\theoremstyle{plain}
\newtheorem*{assumption*}{\protect\assumptionname}
\theoremstyle{plain}
\newtheorem{assumption}{\protect\assumptionname}
\theoremstyle{plain}
\newtheorem{thm}{\protect\theoremname}[section]
\newtheorem{defn}[thm]{\protect\definitionname}
\theoremstyle{definition}
\newtheorem{rem}[thm]{\protect\remarkname}
\theoremstyle{plain}
\newtheorem{prop}[thm]{\protect\propositionname}
\theoremstyle{plain}
\newtheorem{lem}[thm]{\protect\lemmaname}
\theoremstyle{plain}

\usepackage{babel}

\providecommand{\assumptionname}{Assumption}

\providecommand{\definitionname}{Definition}
\providecommand{\lemmaname}{Lemma}
\providecommand{\propositionname}{Proposition}
\providecommand{\remarkname}{Remark}
\providecommand{\theoremname}{Theorem}

\makeatother

\usepackage{babel}
\providecommand{\assumptionname}{Assumption}
\providecommand{\definitionname}{Definition}
\providecommand{\lemmaname}{Lemma}
\providecommand{\propositionname}{Proposition}
\providecommand{\questionname}{Question}
\providecommand{\remarkname}{Remark}
\providecommand{\theoremname}{Theorem}

\author[S. Breteaux]{S{\'e}bastien Breteaux}
\address[S. Breteaux]{Universit{\'e} de Lorraine, CNRS, IECL, F-57000 Metz, France}
\email{sebastien.breteaux@univ-lorraine.fr}

\author[M. Correggi]{Michele Correggi}
\address[M. Correggi]{Dipartimento di Matematica, Politecnico di Milano, P.zza Leonardo da Vinci, 32, 20133, Milano, Italy}
\email{michele.correggi@gmail.com}
\urladdr{https://sites.google.com/view/michele-correggi}

\author[M. Falconi]{Marco Falconi}
\address[M. Falconi]{Dipartimento di Matematica, Politecnico di Milano, P.zza Leonardo da Vinci, 32, 20133, Milano, Italy}
\email{marco.falconi@polimi.it}
\urladdr{https://www.mfmat.org/}

\author[J. Faupin]{J{\'e}r{\'e}my Faupin}
\address[J. Faupin]{Universit{\'e} de Lorraine, CNRS, IECL, F-57000 Metz, France}
\email{jeremy.faupin@univ-lorraine.fr}

\begin{document}

\title[Quantum Point Charges and Quasi-classical Electromagnetic
Fields]{Quantum Point Charges Interacting with Quasi-classical
  Electromagnetic Fields}
\bibliographystyle{natalpha}
\begin{abstract}
  We study effective models describing systems of quantum particles
  interacting with quantized (electromagnetic) fields in the quasi-classical
  regime, i.e., when the field's state shows a large average number of
  excitations. Once the field's degrees of freedom are traced out on
  factorized states, the reduced dynamics of the particles' system is
  described by an effective Schr\"{o}dinger operator keeping track of the
  field's state. We prove that, under suitable assumptions on the latter,
  such effective models are well-posed even if the particles are point-like,
  that is no ultraviolet cut-off is imposed on the interaction with quantum
  fields.
\end{abstract}

\maketitle

\section{Introduction}

It is widely known that models describing quantum particles in interaction
with quantized fields are ill-posed if the former are assumed to be
point-like or, equivalently, no ultraviolet cut-off is imposed on the
interaction with the fields \cite{spohn2004dcp}. A typical and paradigmatic
example is provided by the Nelson model \cite{nelson1964jmp}, where nucleons
are linearly coupled to a scalar quantized field: the ultraviolet
regularization can in this case be removed up to the extraction of an infinite
self-energy and a suitable renormalization procedure \cite{nelson1964jmp}. A
more relevant model is the Pauli-Fierz (PF) Hamiltonian \cite{pauli1938nc},
describing quantum particles interacting with the electromagnetic radiation,
which is well-posed only if the large frequencies of the radiation are
suitably cut off. The removal of such ultraviolet cut-off is one of the major
open problems in non-relativistic Quantum ElectroDynamics (QED) (see,
\emph{e.g.}, \cite[\textsection 19.3]{spohn2004dcp}). In this work, we aim at tackling this
problem in the {\it quasi-classical regime} recently introduced in
\citep{correggi2017ahp,carlone2021sima,correggi2017arxiv,correggi2023jems,correggi2023apde}.

The quasi-classical regime consists of an average number of field's
excitations that is much larger than $ 1 $ (we use natural units in which $
\hbar = 1 $): more precisely, a quasi-classical field state $ \Psi_{\eps} $
satisfies \bdm \lf\langle \mathcal{N} \ri\rangle_{\Psi_\eps} \simeq
\tfrac{1}{\eps}, \qquad \mbox{for } 0 < \eps \ll 1, \edm where \bdm
\mathcal{N} = \int a^{\dagger}(k) a(k) \diff k \edm is the number of the
field's excitations. When this is the case, we can consider the commutator
between the canonical variables $ [a(k), a^{\dagger}(k')] = \delta(k-k') $ to
be negligible and introduce rescaled variables $ a^{\sharp}_{\eps} : =
\sqrt{\eps} a^{\sharp} $, so that \beq
\label{eq: commutation}
\lf[ a_{\eps}(k), a^{\dagger}_{\eps}(k') \ri] = \eps \delta(k-k').  \eeq Such
variables are the ones we are going to use throughout the paper. Their
semiclassical nature is made apparent in the vanishing of the commutator as $
\eps \to 0 $, so that the field can be well approximated by its classical
counterpart.

Let us specify now the setting in more details. We consider a bipartite
quantum system whose space of states is \beq
\label{eq: Hilbert}
\mathscr{H} := L^2(\mathbb{R}^3; \C^{2s+1}) \otimes \mathscr{F}_{\eps}, \eeq
where $ \mathscr{F}_{\eps} $ is a suitable Fock space describing the field's
degrees of freedom and $ s \in \tfrac{1}{2} \N $ stands for the spin of the
particle. We are assuming for simplicity that there is a single quantum
particle interacting with the field but the model can be easily generalized
to many-body systems. The Hamiltonian for the full system is denoted by $
\mathbb{H}_{\eps} $ and contains a non-trivial interaction term, i.e., not
factorized: we aim to address models of non-relativistic QED and therefore
the PF Hamiltonian, but for the sake of providing a simpler and pedagogical
example we will also discuss the Nelson Hamiltonian. However, the general
scheme is independent of the specific details of the Hamiltonian $
\mathbb{H}_{\eps} $.  Our main goal is indeed to study the reduced operators
obtained by tracing out the field's degrees of freedom on a product state of
the form \bdm \psi \otimes \Psi_{\eps} \in \mathscr{H}, \edm where $ \Psi_{\eps} \in
\mathscr{F}_{\eps} $ is a quasi-classical state in the sense specified
above. We thus consider the quadratic form \beq
\label{eq: expectation}
\mathcal{Q}_{\eps}[\psi] : = \meanlr{\psi \otimes \Psi_{\eps}}{\mathbb{H}_{\eps}}{\psi \otimes \Psi_{\eps}}_{\mathscr{H}} -
\meanlr{\psi \otimes \Psi_{\eps}}{\diff \Gamma_{\eps}(\omega)}{\psi \otimes \Psi_{\eps}}_{\mathscr{H}}, \eeq
and study its limit as $ \eps \to 0 $, where $ \diff \Gamma_{\eps}(\omega) $ stands for
the second quantization of the dispersion relation $ \omega $.

More precisely, we are going to show that, while the operator $ \mathbb{H}_{\eps} $ is
in general well defined only in presence of a suitable ultraviolet
regularization, the form $ \mathcal{Q}_{\eps}[\psi] $ is well posed even without such an
ultraviolet cut-off, provided the field's state $ \Psi_{\eps} $ is regular
enough. Then, under the same assumptions on $ \Psi_{\eps} $, we prove that the
quadratic form $ \mathcal{Q}_{\eps} $ converges as $ \eps \to 0 $ to a quadratic form $
\mathcal{Q}_{\mu} $ depending on a classical Wigner measure $\mu$ on the one-excitation
space of the field: \bdm \mathcal{Q}_{\eps} \xrightarrow[\eps \to 0]{\Gamma} \mathcal{Q}_{\mu}, \edm where
the convergence is in the sense of $ \Gamma-$convergence of functionals
\cite{dalmaso1993pnde}. Furthermore, $ \mathcal{Q}_{\mu} $ uniquely identifies a
self-adjoint Schr\"{o}dinger operator $ H_{\mu} $ characterizing the particle's
reduced dynamics. Since (weak and strong) $ \Gamma-$convergence of quadratic forms
is equivalent
\citep[][\textsection13]{dalmaso1993pnde
} to strong resolvent convergence of the associated operators, we deduce that
the reduced particle's dynamics converges as $ \eps \to 0 $ to the one
generated by $ H_{\mu} $. In addition, if the particle is trapped, i.e., a
confining potential is present, then the convergence of the generators is
lifted to norm resolvent sense. We already point out that in the case of the
PF model, in order to prove the above convergence, we will have to perform a
vacuum renormalization and remove some energy diverging as $ \eps \to 0 $, or,
equivalently, consider the normal ordered version of the PF Hamiltonian $
\mathbb{H}_{\eps} $. 
One of the tools we use to handle quadratic forms without ultraviolet cutoff is the use of suitable Lorentz spaces, in a similar fashion as in the works \cite{BFP1,BFP2}.

\medskip

\paragraph{\textbf{Organisation of the paper}}We present the statements of the main results in  \cref{sec:main-results}.  \cref{sec:preliminaries} is devoted to some functional inequalities in Lorentz spaces along with the semiclassical analysis framework used along the paper.  \cref{sec:Nelson} presents the proofs of our results for the Nelson model while  \cref{sec:Pauli-Fierz} contains the proofs of our results for the Pauli-Fierz model. Finally, in  \cref{sec:Gamma-conv}, we prove a $\Gamma$-convergence result which we need in the core of the article.

\section{Main Results}\label{sec:main-results}

We present here our main results. In order to provide a precise statement we
have first to address the notion of convergence in the quasi-classical limit
and provide a definition of the Wigner measures associated to (families of)
field states $ \Psi_{\eps} $. Then, we first state the results concerning the Nelson
model and next discuss non-relativistic QED.

We recall that the Hilbert space on which $ \mathbb{H}_{\eps} $ acts is $ L^2(\mathbb{R}^3;
\C^{2s+1}) \otimes \mathscr{F}_{\eps} $, where $ \mathscr{F}_{\eps} $ is the
symmetric Fock space constructed over the one-excitation space $ \mathfrak{h}^{\sharp} $,
i.e., 
\begin{equation}
\label{eq: fock}
\mathscr{F}_{\eps}=\Gamma_{\varepsilon}(\mathfrak{h}^{\sharp}) = \bigoplus_{n \in \N} {\mathfrak{h}^{\sharp}}^{\otimes_{\mathrm{s}} n},  
\end{equation}
where
\begin{equation}
  \mathfrak{h}^{\mathrm{Nel}} : = L^2(\R^3),	\qquad		\mathfrak{h}^{\mathrm{PF}} : = L^2(\R^3; \C^2).	
\end{equation}
The Hamiltonian of the complete system is assumed to have the formal
structure \beq \label{eq:defH}\mathbb{H}^{\sharp}_{\eps} = H^{\sharp}_0 + \diff \Gamma_{\eps}(\omega) + \mathbb{H}^{\sharp}_{I, \eps},
\eeq where $ H^{\sharp}_0 $ is the particle's Hamiltonian acting non-trivially
only on $ L^2(\R^3; \C^{2s+1}) $, $ \Gamma_{\eps} $ the second quantization map
with canonical variables satisfying \eqref{eq: commutation}, $ \omega $ the
field's dispersion relation and $ \mathbb{H}^{\sharp}_{I, \eps} $ a non-factorized
operators describing the particle-field interaction. Here and in the sequel, in the case where $\mathfrak{h}^\sharp=\mathfrak{h}^{\mathrm{Nel}}$, we identify
a map $\omega^{\alpha} :\mathbb{R}^3\to\mathbb{R}$, for $ \alpha \in \R $, with the operator of multiplication by $\omega^{\alpha}$ on
$L^2(\mathbb{R}^3)$. Likewise, in the case where $\mathfrak{h}^\sharp=\mathfrak{h}^{\mathrm{PF}}$, we identify
$\omega^{\alpha}$ with the operator $\omega^{\alpha} \:\mathds{1}_{\mathbb{C}^2}$ on
$L^2(\mathbb{R}^3;\mathbb{C}^2)$. We then recall that the operator $\diff\Gamma_\varepsilon(\omega^{\alpha})$ in
\eqref{eq:defH} is defined by
\begin{equation*}
	\lf. \mathrm{d}\Gamma_{\varepsilon}(\omega^{\alpha}) \ri|_{\mathfrak{h^{\sharp}}^{\otimes_{\mathrm{s}} n}}=\varepsilon\sum_{k=1}^n\mathds{1}_{\mathfrak{h^{\sharp}}^{\otimes k-1}}\otimes \omega^{\alpha} \otimes \mathds{1}_{\mathfrak{h^{\sharp}}^{\otimes n-k}}.
\end{equation*}

Concerning the dispersion relation $ \omega $, we are going to assume the
following properties, which are satisfied by the typical choices $ \omega(k) = |k|
$ or $ \omega(k) = \sqrt{k^2 + m^2} $.

\renewcommand{\theassumption}{$(\mathrm{A}_{\omega})$}
\begin{assumption}
  \label{hyp:A-omega}
  The map $\omega:\mathbb{R}^{3}\to\mathbb{R}_{+}$ is measurable, it grows at least linearly, i.e., 
  	\beq
  		  \liminf_{|k| \to +\infty} \frac{\omega(k)}{|k|} > 0,
	\eeq
	and it admits an (unbounded) inverse $
  \omega^{-1} $ with dense domain $ \mathscr{D}(\omega^{-1}) \subset \mathfrak{h}^{\sharp} $.
\end{assumption}
 
\subsection{Quasi-classical limit}

We recall some facts on semiclassical measures (see,
\emph{e.g.}, \cite{ammari2008ahp,AmmariNier09,AmmariNier11,AmmariNier15,falconi2017arxiv}). Analogously to the notation
introduced above, $ \mathfrak{h}^{\sharp}_{\omega} $ denote the one-excitation spaces
\begin{equation}
  \mathfrak{h}^{\mathrm{Nel}}_{\omega} : = L^2(\R^3, \omega(k) \diff k),	\qquad		\mathfrak{h}^{\mathrm{PF}}_{\omega} : = L^2(\R^3, \omega(k) \diff k; \C^2).	
\end{equation}
More generally, we define the spaces
$\mathfrak{h}^{\sharp}_{\omega^{\alpha}}$ as the weighted $L^2$-spaces with
weight $\omega^{\alpha}$, $\alpha\in \mathbb{R}$. Observe that for all
$\alpha\in \mathbb{R}$, $\mathfrak{h}^{\sharp}\cap
\mathfrak{h}_{\omega^{\alpha}}^{\sharp}$ is dense in both
$\mathfrak{h}^{\sharp}$ and $\mathfrak{h}^{\sharp}_{\omega^{\alpha}}$ thanks
to \cref{hyp:A-omega}. The Weyl operator associated to Nelson and Pauli-Fierz
fields reads
\begin{equation}
  \label{eq: weyl}
  W_{\varepsilon}(z) :=e^{i(a^{\dagger}_{\varepsilon}(z)+a_{\varepsilon}(z))}\; 
\end{equation}
for $ z \in \mathfrak{h}^{\sharp} $. To shorten the notation, in the following we will omit the
label $ {\mathrm{Nel/PF}} $ distinguishing between the Nelson and Pauli-Fierz
models, when the statement applies to both cases. We denote by
$\mathscr{P}(\mathfrak{h})$ the set of Borel probability measures on $\mathfrak{h}$.

\begin{defn}[Semiclassical convergence]
  \label{def:muNelson}
  \mbox{}		\\
  Given a family of normalized microscopic states $ \lf\{ \Psi_{\varepsilon}
  \ri\}_{\varepsilon\in (0,1)} \subset \mathscr{F}_{\eps} $, let us define
  the associated set of quasi-classical
  $\mathfrak{h}_{\omega^{\alpha}}$-Wigner measures $
  \mathscr{M}_{\omega^{\alpha}}(\Psi_{\eps}) \subset
  \mathscr{P}(\mathfrak{h}_{\omega^{\alpha}})$, $\alpha\in \mathbb{R}$, as
  the subset of all probability measures $ \mu $, such that $ \exists \lf\{
  \eps_n \ri\}_{n \in \N} $, $ \eps_n \xrightarrow[n \to +\infty]{} 0 $, so
  that
  \begin{equation}
    \label{eq: convergence}
    \lim_{n\to + \infty} \braketr{\Psi_{\varepsilon_n}}{W_{\varepsilon_n}(\eta) \Psi_{\varepsilon_n}}_{\mathscr{F}_{\eps_n}} = \widehat{\mu}(\eta) : = \int_{\mathfrak{h}_{\omega^{\alpha}}} \mathrm{d}\mu(z) \: e^{2i\Re \braketr{\omega^{-\alpha/2} \eta}{\omega^{\alpha/2} z}_{\mathfrak{h}}} 
  \end{equation}
  for all $\eta\in \mathfrak{h}_{\omega^{-\alpha}}\cap \mathfrak{h}$, which
  we also denote for short as $ \Psi_{\varepsilon_n} \xrightarrow[n \to +
  \infty]{\omega^{\alpha}\mathrm{-sc}} \mu $.
\end{defn}

To ensure that the sequence of field states we are considering admits at
least one limit point, i.e., the set of associated Wigner measures $
\mathscr{M}_{\omega}(\Psi_{\eps}) $ is non-empty, we assume a uniform control on
these states in $\varepsilon$. The precise statement of such a control depends on the
model and therefore we make two different assumptions for the Nelson and PF
models, respectively. We preliminary recall that \bdm \D \Gamma_{\varepsilon}^{(2)}(\omega\otimes\omega)=\D
\Gamma_{\varepsilon}(\omega)^{2}-\varepsilon \,\D \Gamma_{\varepsilon}(\omega^{2}).  \edm

\renewcommand{\theassumption}{$(\mathrm{A}^{\mathrm{Nel}}_{\Psi})$}
\begin{assumption}
  \label{hyp:A-Nel}
  The family $ \lf\{ \Psi_{\varepsilon} \ri\}_{\varepsilon \in (0,1)} \subset
  \Gamma_{\eps}(L^2(\mathbb{R}^3))$ is such that \beq
  \meanlr{\Psi_{\varepsilon}}{1+\,\D
    \Gamma_{\varepsilon}(\omega)}{\Psi_{\varepsilon}}_{\mathscr{F}_{\eps}}
  \leq C \eeq uniformly in $ \varepsilon $.
\end{assumption}

\renewcommand{\theassumption}{$(\mathrm{A}^{\mathrm{PF}}_{\Psi})$}
\begin{assumption}
  \label{hyp:A-PF}
  The family $ \lf\{ \Psi_{\varepsilon} \ri\}_{\varepsilon \in (0,1)} \subset
  \Gamma_{\eps}(L^2(\mathbb{R}^3; \C^2))$ is such that \beq
  \meanlr{\Psi_{\varepsilon}}{1+\,\D \Gamma_{\varepsilon}(\omega) + \D
    \Gamma_{\varepsilon}^{(2)}(\omega\otimes\omega))}{\Psi_{\varepsilon}}_{\mathscr{F}_{\eps}}
  \leq C \eeq uniformly in $ \varepsilon $.
\end{assumption}

We anticipate that for any family of states $ \lf\{ \Psi_{\eps} \ri\}_{\eps}
\subset \mathscr{F}_{\eps} $ as in \cref{def:muNelson}, such that \cref{hyp:A-Nel}
or \cref{hyp:A-PF} holds, we have that $\mathscr{M}_{\omega}(\Psi_{\eps}) \neq \varnothing$
\citep{ammari2008ahp,falconi2017ccm}, i.e., there exists at least one $ \mu \in
\mathscr{M}_{\omega}(\Psi_{\eps}) $, such that \bdm \Psi_{\varepsilon_n} \xrightarrow[n \to + \infty]{\omega
  \mathrm{-sc}} \mu.  \edm In addition, for any such converging sequence
$\lf\{\Psi_{\varepsilon_n} \ri\}_{n\in \mathbb{N}}$ and $ \mu $,
\begin{equation}
  \braketr{\Psi_{\varepsilon_n}}{a_{\varepsilon_n}^{*}(g)\Psi_{\varepsilon_n}}_{\mathscr{F}_{\eps_n}} \xrightarrow[n\to \infty]{} \int_{\mathfrak{h}_{\omega}} \braketr{\omega^{1/2} z}{\omega^{-1/2} g}_{\mathfrak{h}} \, \D \mu(z)\;,
\end{equation}
for all $g\in \mathfrak{h}_{\omega^{-1}}$ (see \cref{prop:semicl-analys-1} below).

\subsection{Nelson model}

We consider a spinless non-relativistic particle linearly coupled to a
quantized scalar field. The space of states \eqref{eq: Hilbert} takes in this
case the following form
\begin{equation}
  \mathscr{H}^{\mathrm{Nel}} = L^2(\mathbb{R}^3)\otimes \mathscr{F}^{\mathrm{Nel}}_{\eps} =L^2(\mathbb{R}^3)\otimes \Gamma_{\eps}(\mathfrak{h}^{\mathrm{Nel}}),	\qquad		\mathfrak{h}^{\mathrm{Nel}} = L^2(\R^3).
\end{equation}
The energy of the total quantum system is described by a Nelson-type
Hamiltonian, given by
\begin{equation}
  \mathbb{H}_{\varepsilon}^{\mathrm{Nel}}= \vec{P}^2 + U(x) + \mathrm{d}\Gamma_\varepsilon(\omega) + \Phi_\varepsilon( e^{i x \cdot k} \omega^{-\frac12}\chi ).
\end{equation}
Here $ \vec{P} :=-i{\nabla}_x$ is the momentum\footnote{We use boldface
  letters to denote vectors in $ \R^3 $, whenever we need to stress the
  vector nature of the object.} of the electron,
$U:\mathbb{R}^3\to\mathbb{R}$ is a real external potential, and
\begin{equation*}
  \Phi_\varepsilon(e^{ix\cdot k} \omega^{-\frac12}\chi) := a^*_\varepsilon(e^{ix\cdot k} \omega^{-\frac12}\chi)+a_\varepsilon(e^{ix\cdot k} \omega^{-\frac12}\chi)=\int_{\mathbb{R}^3} \frac{\chi(k)}{\omega^{\frac12}(k)}\big(e^{ix\cdot k}a^*_\varepsilon(k)+e^{-ix\cdot k}a_\varepsilon(k)\big)\mathrm{d}k,
\end{equation*}
is the field operator corresponding to the interaction between the
non-relativistic particle and the field. The function
$\chi:\mathbb{R}^3\to\mathbb{R}$ is an ultraviolet cut-off, which might be
required for the Hamiltonian to identify a self-adjoint operator in the
Hilbert space $\mathscr{H}^{\mathrm{Nel}}$, but which, in our paper, can
subsequently be put equal to $1$.

Before stating the assumptions on $U$ and $\chi$, let us denote by
$U_+:=\max(U,0)$ and $U_-:=\max(-U,0)$ the positive and negative parts
of $U$, respectively.

\renewcommand{\theassumption}{$(\mathrm{A}_{U})$}
\begin{assumption}
  \label{hyp:A-U}
  The potential $U:\mathbb{R}^{3}\to\mathbb{R}$ is such that $U_{+} \in
  L_{\mathrm{loc}}^{1}(\R^3) $ and $ U_{-} $ is KLMN form-bounded w.r.t. $ -
  \Delta $, i.e., there exists $ a \in (0,1) $ and $ b \in \R $ such that
  \beq \meanlr{\psi}{U_-}{\psi} \leq a \meanlr{\psi}{-\Delta}{\psi} + b \lf\|
  \psi \ri\|^2, \qquad \forall \psi \in H^1(\R^3).  \eeq
\end{assumption}

Concerning the ultraviolet cut-off function, we preliminary recall the
definition of weak $ L^p $ spaces: for $1\le p<\infty$, we denote by $
L^{p,\infty}(\mathbb{R}^{3})$ the set of (equivalence classes of) measurable
functions $f:\mathbb{R}^{3}\to\mathbb{C}$ such that the quasi-norm \beq \lf\|
f \ri\|_{L^{p,\infty}}:=p \lf\| \bigl| \lf\{|f|>t \ri\} \bigr|^{1/p}\,t
\ri\|_{L^{\infty}((0,\infty),\D t/t)} \eeq is finite, where, for any
measurable set $ S $, $ |S| $ stands for its Lebesgue measure.

\renewcommand{\theassumption}{$(\mathrm{A}_{\chi})$}
\begin{assumption}
  \label{hyp:A-chi}
  The function $\chi:\mathbb{R}^{3}\to\mathbb{R}$ is such that $\chi/\omega
  \in L^{3,\infty}(\R^3)$.
\end{assumption}

We remark that the function $ \chi \equiv 1 $ can be easily seen to satisfy
the above \cref{hyp:A-chi}, at least in the physically relevant cases $
\omega(k) = \sqrt{k^2+m^2} $ and $ \omega(k) = \lvert k  \rvert_{}^{} $, since \bdm \lf\| 1/\omega
\ri\|_{L^{3,\infty}(\R^3)} \leq 3 \lf\| \bigl| \lf\{1/|k| > t \ri\}
\bigr|^{1/3}\,t \ri\|_{L^{\infty}((0,\infty),\D t/t)} = 3^{2/3} (4\pi)^{1/3}.
\edm More generally, it can be readily seen that $ 1/\omega \in
L^{3,\infty}(\R^3) $ for any $ \omega $ satisfying \cref{hyp:A-omega}. Hence,
\cref{hyp:A-chi} does not require a decay for large $ |k| $ of the function $
\chi $, as it occurs for ultraviolet cut-offs.

We will see in \cref{sec:Nelson} that the following quadratic form on the Schwartz space $\mathcal{S}(\mathbb{R}^3)$ indeed defines a function $V_\varepsilon(x)$ belonging to some Lorentz space
\begin{equation}
  \label{eq:def-Vepsilon}
\int_{\mathbb{R}^3} V_\varepsilon |\psi|^2  := \int_{\mathbb{R}^3} 2\Re \Big( (2\pi)^{3/2} \overline{\mathcal{F}(|\psi|^2)(k)} \frac{\chi(k)}{\sqrt{\omega(k)}} \langle \Psi_\varepsilon | a_\varepsilon^*(k)\Psi_\varepsilon \rangle_{\mathscr{F}_\varepsilon} \Big)\, \mathrm{d}k\,,\quad \forall\psi \in \mathcal{S}(\mathbb{R}^3)\,,
\end{equation}
where  $ \mathcal{F} $ denotes the Fourier transform on $ L^2(\R^d) $, i.e.,
\begin{equation} \label{eqn:Fourier}
	\mathcal{F}(f)(k) : = \frac{1}{(2\pi)^{d/2}} \int_{\R^d} e^{- i k \cdot x} f(x)\, \diff x\, . 
\end{equation}
If both ${\chi}/{\omega}$ and ${\chi}/{\sqrt{\omega}}$ are in $L^2_k(\R^3)$ then the following explicit expression holds
\begin{equation}
   V_{\varepsilon}(x) = \meanlr{\Psi_{\varepsilon}}{2\mathrm{Re}\:a_{\varepsilon}^{*}(e^{ix\cdot k}\tfrac{\chi}{\sqrt{\omega}})}{\Psi_{\varepsilon}}_{\mathscr{F}_{\eps}}\,,
\end{equation}
so that the expectation of $ \mathbb{H}_{\eps}^{\mathrm{Nel}} - \diff
\Gamma_{\eps}(\omega) $ on the factorized state $ \psi \otimes \Psi_{\eps} $
as in \eqref{eq: expectation} reads
\begin{equation}
  Q_{\eps}[\psi] : = \meanlr{\psi}{\Hnele}{\psi}_{L^2_x},	\qquad	\Hnele=-\Delta+U+V_{\varepsilon}.
\end{equation}
Finally, given $\mu\in \mathscr{P}(\mathfrak{h}^{\mathrm{Nel}}_{\omega})$, we will also see in \cref{sec:Nelson} that the following quadratic form on the Schwartz space $\mathcal{S}(\mathbb{R}^3)$ defines a function $V_\mu$ belonging to some Lorentz space
\begin{equation}
  \label{eq:def-Vmu}
\int_{\mathbb{R}^3} V_\mu |\psi|^2 := \int_{\mathfrak{h}_\omega} 2\Re  \Big\langle  \sqrt{\omega}\,z\Big|  (2\pi)^{3/2} \overline{\mathcal{F}(|\psi|^2)} \frac{\chi}{\omega} \Big\rangle_{L^2_k} \,\mathrm{d}\mu(z) \,,\quad \forall\psi \in \mathcal{S}(\mathbb{R}^3)\,.
\end{equation}
If ${\chi}/{\omega}$ is in $L^2_k(\R^3)$ then the following explicit expression holds
\begin{equation}
  V_{\mu}(x)= 2 \mathrm{Re}\,\int_{\mathfrak{h}_{\omega}} \left\langle \omega^{1/2}z \right \vert\left . \tfrac{\chi}{\omega}e^{ik\cdot x} \right\rangle_{L^2_k} \, \D  \mu(z) \, .
\end{equation}

Our main result for the Nelson model is the following. Let us preliminarily
introduce a slightly modified notion of $\Gamma$-convergence, adapted to our
needs.
\begin{defn}
Given $\mathcal{E}_n$, with $n\in\mathbb{N}\cup \{\infty\}$, functionals from  $\mathfrak{h}^{\sharp}$ to $\mathbb{R}$ defined on a suitable common dense domain $\mathscr{Q}\subset\mathfrak{h}^{\sharp}$ we write
$$\mathcal{E}_{n}[\cdot ]\xrightarrow[n\to \infty]{\Gamma} \mathcal{E}_\infty[\cdot ]\;,$$
if and only if the following two statements hold:
\begin{itemize}
\item {$\mathrm{[}\Gamma-\limsup\mathrm{]}$} For any $\psi\in \mathscr{Q}$, there exists one sequence
  $\{\psi_n\}_{n\in \mathbb{N}} \subset \mathscr{Q}$ such that $\psi_n\xrightarrow[n\to \infty]{\mathfrak{h}^{\sharp}}\psi$,
  and
  \begin{equation*}
    \limsup_{n\to \infty}\mathcal{E}_n[\psi_n]\leq \mathcal{E}_\infty[\psi]\;;
  \end{equation*}
\item {$\mathrm{[}\Gamma-\liminf\mathrm{]}$} For any sequence $\{\psi_n\}_{n\in \mathbb{N}}\subset \mathscr{Q}$, such
  that $\psi_n\xrightarrow[n\to \infty]{\mathrm{w-}\mathfrak{h}^{\sharp}}\psi$, $\psi\in \mathscr{Q}$,
  \begin{equation*}
    \liminf_{n\to \infty}\mathcal{E}_n[\psi_n]\geq  \mathcal{E}_\infty[\psi]\;.
  \end{equation*}
\end{itemize}

\end{defn}
The convergence~$\mathcal{E}_{n}[\cdot ]\xrightarrow[n\to \infty]{\Gamma} \mathcal{E}_\infty[\cdot ]$ holds if and only if $\mathcal{E}_{n}$ $\Gamma$-converges to~$\mathcal{E}_\infty$ both in the weak and strong topologies of~$\mathfrak{h}^{\sharp}$.

Let us recall that with this definition, the $\Gamma$-convergence of quadratic
forms bounded from below is equivalent to strong resolvent convergence of the
associated self-adjoint operators \citep[][Theorem
13.6]{dalmaso1993pnde
}. If, in addition, the operators have compact resolvent, then the
$\Gamma$-convergence is equivalent to norm resolvent convergence.

\begin{thm}[Convergence of $ H^{\mathrm{Nel}}_{\varepsilon}$]
  \label{thm:Quasiclassical-limit-Nelson}
  \mbox{}		\\
  Suppose that \cref{hyp:A-omega,hyp:A-chi,hyp:A-U,hyp:A-Nel} hold. For any $\mu\in
  \mathscr{M}^{\mathrm{Nel}}(\Psi_\varepsilon)$ and for any sequence
  $\{\varepsilon_n\}_{n\in \mathbb{N}}$, $\varepsilon_n\to 0$, such that
  $\Psi_{\eps_n} \xrightarrow[n \to+\infty]{\mathrm{sc}} \mu$, then, \beq
  H^{\mathrm{Nel}}_{\varepsilon_n}=-\Delta+U+V_{\varepsilon_n}\qquad\text{and}\qquad
  \Hnelm=-\Delta+U+V_{\mu} \eeq define symmetric closed quadratic forms with
  form
  domain $$\mathcal{Q}:=\mathcal{Q}(H^{\mathrm{Nel}}_{\varepsilon_n})=\mathcal{Q}(\Hnelm)=H^{1}(\R^3)
  \cap L^{2}(\R^3, U_{+}\, \D x)$$ and hence define self-adjoint operators on
  domains $\mathcal{D}(H^{\mathrm{Nel}}_{\varepsilon_n}),$ $\mathcal{D}(\Hnelm)\subseteq
  \mathcal{Q} $, respectively. Moreover, as quadratic forms, \beq
  \meanlr{\varphi}{H^{\mathrm{Nel}}_{\varepsilon_n}}{\varphi} \xrightarrow[n\to
  \infty]{\Gamma} \meanlr{\varphi}{\Hnelm}{\varphi}\,.  \eeq Consequently,
  $H^{\mathrm{Nel}}_{\varepsilon_n}$ converges to $\Hnelm$ in strong resolvent sense.  If
  in addition $U$ is confining, i.e. $U_{+}(x)\to\infty$ as $|x|\to\infty$,
  then $H^{\mathrm{Nel}}_{\varepsilon_n}$ converges to $\Hnelm$ in norm resolvent sense.
\end{thm}

\begin{rem}[Ultraviolet renormalization]
  \label{rem:renormalization}
  \mbox{}	\\
  In the work \citep{correggi2024inprep} it is shown that the ultraviolet
  renormalization of the Nelson model commutes with the quasi-classical
  limit. However, the former calls for the extraction of an infinite particle
  self-energy, and the introduction of a suitable dressing transformation
  that modifies substantially the properties of the microscopic Hamiltonian
  as well as of its quasi-classical limit. In this framework, the above
  result shows that, on product states and at the
  level of the quadratic form, such a renormalization procedure is actually
  not needed. Indeed, one can easily figure out that
  \cref{thm:Quasiclassical-limit-Nelson} entails that (see \cref{hyp:A-chi}
  and discussion thereafter), if one imposes an ultraviolet cut-off $ \chi_{\Lambda} $
  such that $ \chi_{ \Lambda} \to 1 $, as $ \Lambda \to + \infty $, then, at the level of the reduced
  quadratic form, the limits $ \eps \to 0 $ and $ \Lambda \to +\infty $ yield the same
  result irrespective of the order in which they are taken. Furthermore, if
  the cut-off parameter $ \Lambda = \Lambda(\eps) $ depends on $ \eps $ and $ \Lambda(\eps) \to +
  \infty $ as $ \eps \to 0 $, i.e., the ultraviolet renormalization is performed at
  the same time as the quasi-classical limit, the rate of the former does not
  matter.
\end{rem}

\begin{rem}[Wigner measures]
  \label{rem:1}
  \mbox{}	\\
  As discussed above, the set of Wigner measures of a generic family of
  states $\{\Psi_{\varepsilon}\}_{\varepsilon\in (0,1)}$ might contain more
  than a single point. In that case, the family of operators
  $H^{\mathrm{Nel}}_{\varepsilon_n}$ depends on the choice of sequence
  $\{\varepsilon_n\}_{n\in \mathbb{N}}$ determining the limit point
  $\mu$. However, if the set of Wigner measures consists of a single point,
  then the $\Gamma$-convergence holds as $\varepsilon\to 0$.
\end{rem}

\begin{rem}[Pseudo-relativistic kinetic energy]
  \label{rem:pseudo}
  \mbox{}	\\
  A close inspection of the proof shows that the result easily extends to
  Nelson-type Hamiltonians with pseudo-relativistic kinetic energy, i.e., for
  $ H_0=\sqrt{ - \Delta + \nu^2} +U$, $ \nu \geq 0 $, up to a straightforward modification of \cref{hyp:A-U}.
\end{rem}


\subsection{Pauli-Fierz model}

We next consider a second, physically more relevant setting, where the spin
of the non-relativistic particle is taken into account, the radiation is
described by the quantized electromagnetic field in the Coulomb gauge and the
particle-field interaction is given by Pauli coupling. For the sake of simplicity we set $ s = \frac{1}{2} $ but the extension to different values of the spin is straightforward. The space states is
then
\begin{equation}
  \mathscr{H}^{\mathrm{PF}} = L^2(\mathbb{R}^3;\mathbb{C}^2)\otimes \Gamma_{\eps}(\mathfrak{h}),	\qquad	\mathfrak{h}^{\mathrm{PF}} = L^2(\mathbb{R}^3;\mathbb{C}^2), 
\end{equation}
and the Hamiltonian is given by
\begin{equation}
  \label{eq: Ham PF}
  \mathbb{H}^{\mathrm{PF}}_{\varepsilon}=\big(\sigmav\cdot(\vec{P}-\vec{\mathbb{A}}_{\varepsilon}(x))\big)^{2}+U(x) + \mathrm{d}\Gamma_\varepsilon(\omega),
\end{equation}
where the vector of Pauli matrices $\sigmav=(\sigma_1,\sigma_2,\sigma_3)$ is
\begin{equation*}
  \sigma_{1}=\begin{pmatrix}0 & 1\\
    1 & 0
  \end{pmatrix}\,,\;\sigma_{2}=\begin{pmatrix}0 & -i\\
    i & 0
  \end{pmatrix}\,,\;\sigma_{3}=\begin{pmatrix}1 & 0\\
    0 & -1
  \end{pmatrix}\,.
\end{equation*}
The vector potential of the quantized electromagnetic field is of the form
\begin{equation*}
  \vec{\mathbb{A}}_{\varepsilon}(x)=\sum_{\lambda=1,2}\int_{\mathbb{R}^3} \frac{\chi(k)}{\omega^{\frac12}(k)} \vec{e}_{\lambda}(k) \lf( e^{ix\cdot k}a^\dagger_{\lambda,\varepsilon}(k)+e^{-ix\cdot k}a_{\lambda,\varepsilon}(k) \ri)\mathrm{d}k,
\end{equation*}
with $(\vec{e}_1(k),\vec{e}_2(k))$ polarization vectors, such that
$(\vec{e}_1(k),\vec{e}_2(k),k/|k|)$ forms an orthonormal basis of
$\mathbb{R}^3$ for all $k\neq0$. The creation and annihilation operators
$a^{\dagger}_{\lambda,\varepsilon}$, $a_{\lambda,\varepsilon}$ are now labelled by the
polarization directions $ \lambda, \lambda' \in \lf\{1,2\ri\} $ and satisfy
the canonical commutation relations
\begin{equation*} \lf[ a^\sharp_{\lambda,\varepsilon}(k) ,
  a^\sharp_{\lambda',\varepsilon}(k') \ri] = 0 , \qquad \lf[
  a_{\lambda,\varepsilon}(k) , a^\dagger_{\lambda',\varepsilon}(k') \ri] =
  \varepsilon \delta_{\lambda\lambda'}\delta( k - k' ),
\end{equation*}
for $k,k'\in\mathbb{R}^3$ and $\lambda,\lambda'\in\{1,2\}$. To shorten
notations, we will sometimes write \bdm
\vec{\mathbb{A}}_{\varepsilon}(x)=a^*_\varepsilon(\vec{w}_x)+
a_\varepsilon(\vec{w}_x), \edm with
\begin{equation}\label{eq:defwx}
  \vec{w}_x(k,\lambda):=\frac{\chi(k)}{\omega^{\frac12}(k)}\vec{e}_{\lambda}(k)e^{ix\cdot k}.
\end{equation}

The Hamiltonian \eqref{eq: Ham PF} is not normal ordered, which implies that
its vacuum energy may diverge in absence of an ultraviolet cut-off. For this
reason, we actually work with its Wick-ordered counterpart \bml{
  :\mathbb{H}_\varepsilon: \, = (\sigmav\cdot\vec{P})^{2}-(\sigmav\cdot\vec{P})(\sigmav\cdot\vec{\mathbb{A}}_{\varepsilon}(x))-(\sigmav\cdot\vec{\mathbb{A}}_{\varepsilon}(x))(\sigmav\cdot\vec{P})\\
  +a^*_\varepsilon(\vec{w}_x)a^*_\varepsilon(\vec{w}_x)+a_\varepsilon(\vec{w}_x)a_\varepsilon(\vec{w}_x)+2a^*_\varepsilon(\vec{w}_x)a_\varepsilon(\vec{w}_x)+
  U(x).  } Note that $:\mathbb{H}_\varepsilon:$ and $\mathbb{H}_\varepsilon$
only differ by an $\varepsilon$-dependent constant,
\begin{equation*}
  :\mathbb{H}_\varepsilon: \, = \mathbb{H}_\varepsilon - \lf[a_\varepsilon(\vec{w}_x),a^*_\varepsilon(\vec{w}_x) \ri] = \mathbb{H}_\varepsilon - 2\varepsilon \lf\|\omega^{-1/2}\chi \ri\|^2_{L^2_k},
\end{equation*}
which formally vanishes in the limit $ \eps \to 0
$. 

%

The counterparts of the potential $ V_{\eps} $ defined in \eqref{eq:def-Vepsilon}  for the PF model are a vector potential~$ \vec{A}_\varepsilon$ along with a potential $W_\varepsilon$ defined through the following quadratic form on the Schwartz space $\mathcal{S}(\mathbb{R}^3)$. For $ \vec{A}_\varepsilon$, with $\vec{w}_x$ is defined in \eqref{eq:defwx}:
\begin{equation}\label{eq:def-Aepsilon}
\int_{\mathbb{R}^3}|\psi|^2 \vec{A}_\varepsilon := \int_{\mathbb{R}^6} 2\Re \Big( |\psi(x)|^2 \ \vec{w}_x(k) \ \langle \Psi_\varepsilon | a_\varepsilon^*(k)\Psi_\varepsilon \rangle_{\mathscr{F}_\varepsilon} \Big)\,\mathrm{d}x \, \mathrm{d}k\,,\quad \forall\psi \in \mathcal{S}(\mathbb{R}^3)\,.
\end{equation}
If both ${\chi}/{\omega}$ and ${\chi}/{\sqrt{\omega}}$ are in $L^2_k(\R^3)$ then the following explicit expression holds
\beq 
   \vec{A}_{\varepsilon}(x) := \braketr{ \Psi_{\varepsilon}}{\vec{\mathbb{A}}_{\varepsilon}(x)\: \Psi_{\varepsilon}}_{\mathscr{F}_{\eps}}
  =
  \braketr{\Psi_{\varepsilon}}{\lf(a^\dagger_\varepsilon(\vec{w}_x)+a_\varepsilon(\vec{w}_x)\ri)\:\Psi_{\varepsilon}}_{\mathscr{F}_{\eps}}\,.
\eeq
For $W_\varepsilon$:
\begin{multline}\label{eq:def-Wepsilon}
\int_{\mathbb{R}^3}|\psi|^2 \ W_\varepsilon := \int_{\mathbb{R}^9}   |\psi(x)|^2 \ 2\Re \Big( \vec{w}_x(k)\cdot\vec{w}_x(k') \ \langle a_\varepsilon(k)a_\varepsilon(k') \Psi_\varepsilon | \Psi_\varepsilon \rangle_{\mathscr{F}_\varepsilon} \\
+ \vec{w}_x(k)\cdot \overline{\vec{w}_x(k')}\langle a_\varepsilon(k)\Psi_\varepsilon | a_\varepsilon(k')\Psi_\varepsilon \rangle_{\mathscr{F}_\varepsilon} \Big)\,\mathrm{d}x \, \mathrm{d}k \, \mathrm{d}k'\,,\quad \forall\psi \in \mathcal{S}(\mathbb{R}^3)\,.
\end{multline}
If both ${\chi}/{\omega}$ and ${\chi}/{\sqrt{\omega}}$ are in $L^2_k(\R^3)$ then the following explicit expression holds
\begin{multline}
   W_{\varepsilon}(x)  := \braketr{ \Psi_{\varepsilon}}{ :\big(\vec{\mathbb{A}}_{\varepsilon}(x)\big)^{2}: \Psi_{\varepsilon}}_{\mathscr{F}_{\eps}} \\
  = \braketr{ \Psi_{\varepsilon}}{\lf( a^\dagger_\varepsilon(\vec{w}_x) a^\dagger_\varepsilon(\vec{w}_x) + a_\varepsilon(\vec{w}_x)a_\varepsilon(\vec{w}_x) + 2 a^\dagger_\varepsilon(\vec{w}_x)a_\varepsilon(\vec{w}_x)\ri)  \Psi_{\varepsilon}}_{\mathscr{F}_{\eps}}\,.
\end{multline}
The reduced Hamiltonian has
then the following expression: \beq H^{\mathrm{PF}}_{\varepsilon}
=(\sigmav\cdot\vec{P})^{2}-(\sigmav\cdot\vec{P})(\sigmav\cdot\vec{A}_{\varepsilon}(x))-(\sigmav\cdot\vec{A}_{\varepsilon}(x))(\sigmav\cdot\vec{P})+W_{\varepsilon}(x)+U(x).
\eeq Its quasi-classical counterpart is identified by a measure $ \mu \in
\mathscr{M}^{\mathrm{PF}}(\Psi_\varepsilon)$ in the set of Wigner measures
associated to $ \lf\{ \Psi_{\eps} \ri\}_{\eps \in (0,1)} $ and such that, up to the extraction
of a subsequence, \bdm \Psi_{\eps_n} \xrightarrow[n \to+\infty]{\mathrm{sc}}
\mu, \edm and reads
\begin{equation}\label{eq:def-Hmu-PF}
  H^{\mathrm{PF}}_{\mu}:=(\sigmav\cdot\vec{P})^{2}-(\sigmav\cdot\vec{P})(\sigmav\cdot\vec{A}_{\mu}(x))-(\sigmav\cdot\vec{A}_{\mu}(x))(\sigmav\cdot\vec{P})+W_{\mu}(x)+U(x),
\end{equation}
where the multiplication operators by $\vec{A}_{\mu}(x)$ and $W_\mu(x)$ are defined through the following quadratic forms on $\mathcal{S}(\mathbb{R}^3)$ (see \cref{sec:Pauli-Fierz}):
\begin{equation}\label{eq:def-Amu}
\int_{\mathbb{R}^3} \mathbf{A}_\mu \ |\psi|^2 := \int_{\mathfrak{h}^{\mathrm{PF}}} 2\mathrm{Re}\Big( \int_{\mathbb{R}^6} \overline{z(k)} \ \mathbf{w}_x(k) \ |\psi(x)|^2  \,\mathrm{d}x \,\mathrm{d}k \Big)\mathrm{d}\mu(z)
\end{equation}
and
\begin{equation}\label{eq:def-Wmu}
  \int_{\mathbb{R}^3} W_{\mu}\ |\psi|^2 := \int_{\mathfrak{h}^{\mathrm{PF}}_{\omega}}  \int_{\mathbb{R}^9}  4\mathrm{Re}(\overline{z(k_1)} \ \vec{w}_{x}(k_1)) \mathrm{Re}(\overline{z(k_2)} \ \vec{w}_{x}(k_2))  |\psi(x)|^2\ \mathrm{d}x \ \mathrm{d}k_1\ \mathrm{d}k_2  \ \diff\mu(z)\, .
\end{equation}
with $\psi\in\mathcal{S}(\mathbb{R}^3)$.
In the simple case that ${\chi}/{\sqrt{\omega}}$ is in $L^2_k(\R^3)$ then the following explicit expressions hold
\begin{equation}
  \vec{A}_{\mu}(x):=\int_{\mathfrak{h}_{\omega}^{\mathrm{PF}}}2\mathrm{Re}\, \braketr{ \omega^{1/2} z}{\omega^{-1/2}\vec{w}_{x}}_{L^2_k} \diff\mu(z)\, ,
\end{equation}
\begin{equation}
  W_{\mu}(x):=\int_{\mathfrak{h}^{\mathrm{PF}}_{\omega}} \lf( 2\mathrm{Re}\, \braketr{ \omega^{1/2}z}{\omega^{-1/2}\vec{w}_{x}}_{L^2_k} \ri)^{2} \diff\mu(z)\, .
\end{equation}

To prove our result on the Pauli-Fierz model, we need to slightly strengthen \cref{hyp:A-chi}:
\renewcommand{\theassumption}{$(\mathrm{A}'_{\chi})$}
\begin{assumption}
  \label{hyp:A'-chi}
  The function $\chi:\mathbb{R}^{3}\to\mathbb{R}$ is such that $|k|\chi/\omega
  \in L^{\infty}(\R^3)$.
\end{assumption}
Similarly as for \cref{hyp:A-chi}, the function $\chi\equiv 1$ satisfies \cref{hyp:A'-chi} in the physically relevant cases $
\omega(k) = \sqrt{k^2+m^2} $ and $ \omega(k) = \lvert k  \rvert_{}^{} $.

Our main result is then:
\begin{thm}[Convergence of $H^{\mathrm{PF}}_{\varepsilon}$]
  \label{thm:Quasi-classical-limit-Pauli-Fierz}
  \mbox{}	\\
  Suppose that \cref{hyp:A-omega,hyp:A'-chi,hyp:A-U,hyp:A-PF} hold. Then, for any
  $\mu\in \mathscr{M}^{\mathrm{PF}}(\Psi_\varepsilon)$ and for any sequence $\{\varepsilon_n\}_{n\in \mathbb{N}}$,
  $\varepsilon_n\to 0$, such that $\Psi_{\eps_n} \xrightarrow[n \to+\infty]{\mathrm{sc}} \mu$,   $H^{\mathrm{PF}}_{\varepsilon_n}$ and $\Hpfm$ define symmetric
  closed quadratic forms with form domain 
  \beq
  	\mathcal{Q}= H^{1}(\R^3) \cap L^{2}(\R^3,
  U_{+}\, \D x; \C^2),
	\eeq
	 and hence define self-adjoint operators on domains
  $\mathcal{D}(H^{\mathrm{PF}}_n),$ $\mathcal{D}(\Hpfm)\subseteq \mathcal{Q} $, respectively. Moreover, as
  quadratic forms,
  \beq
    \meanlr{\varphi}{H^{\mathrm{PF}}_{\varepsilon_n}}{ \varphi} \xrightarrow[n\to
    \infty]{\Gamma} \meanlr{\varphi}{\Hpfm}{\varphi}\,.
  \eeq
  Consequently, $H^{\mathrm{PF}}_{\varepsilon_n}$ converges to $\Hpfm$ in strong resolvent
  sense. If in addition $U$ is confining, i.e. $U_{+}(x)\to\infty$ as
  $|x|\to\infty$, then $H^{\mathrm{PF}}_{\varepsilon_n}$ converges to $\Hpfm$ in norm
  resolvent sense.
\end{thm}

\begin{rem}[Ultraviolet renormalization]
  \label{rem:renormalization2}
  \mbox{}	\\
  Analogously to \cref{rem:renormalization}, the above
  \cref{thm:Quasi-classical-limit-Pauli-Fierz} implies (see \cref{hyp:A-chi}
  and discussion thereafter) that one can remove the ultraviolet cut-off at
  the level of the energy form and such an operation commutes with the
  quasi-classical limit $ \eps \to 0 $. However, unlike for the Nelson model,
  the rate at which the cut-off is removed does matter: recalling that the
  result above applies to the normal ordered form of the PF Hamiltonian, the
  vacuum energy we are implicitly extracting is \bdm - 2\varepsilon \lf\|\omega^{-1/2}\chi_{\Lambda}
  \ri\|^2_{L^2_k} = - C \eps \Lambda^2 (1 + o_{\Lambda}(1)) \edm in the simple case of a
  sharp cut-off $ \chi_{\Lambda} = \one_{|k| \leq \Lambda} $. Hence, such an energy remains
  bounded as $ \eps \to 0 $, if $ \Lambda \lesssim \frac{1}{\sqrt{\eps}} $, suggesting that
  the ultraviolet renormalization is actually trivial if the cut-off is
  removed slowly enough.
\end{rem}
\begin{rem}[Quantum and classical divergences in non-relativistic electrodynamics]
  \label{rem:qcdivergences}
  \mbox{}     \\
  The choice $\chi \equiv 1$ corresponds to a point distribution of the electric
  charge of the quantum Schr\"{o}dinger particle. On the one hand, at the quantum
  level, a point charge yields an energy unbounded from below due to the
  interaction with the \emph{quantized} field. On the other hand, in
  classical electrodynamics, a point charge yields an energy unbounded from
  below due to the Coulomb singularity of the electrostatic potential. The
  model in between, namely with a quantum point charge and a classical
  electromagnetic field has an energy bounded from below, at least for
  regular enough vector potentials of the electromagnetic
  field. \cref{thm:Quasi-classical-limit-Pauli-Fierz} agrees with such a
  physical picture, since we can derive the Hamiltonian of a point charge in
  a classical electromagnetic field from a fully quantized non-relativistic
  electrodynamics, under suitable assumptions.
\end{rem}

\section{Preliminaries}\label{sec:preliminaries}

\subsection{Functional inequalities in Lorentz spaces}

Our proofs rely on suitable functional inequalities in Lorentz spaces. For
the convenience of the reader we recall basic facts about Lorentz spaces
(see, \emph{e.g.}, \cite{ONeil63,Yap69,Lemarie-Rieusset02,BezLeeNakamuraSawano17} or \cite[1.4.19]{Grafakos08} for more details).
For $1\le p<\infty$ and $1\le
q\leq\infty$, the Lorentz spaces $L^{p,q} :=L^{p,q}(\mathbb{R}^{d})$ are
defined as the set of (equivalence classes of) measurable functions
$f:\mathbb{R}^{d}\to\mathbb{C}$ such that the quasi-norm
\beq
  \lf\|f\ri\|_{L^{p,q}}:=p^{1/q} \lf\| \lf| \{|f|>t\} \ri|^{1/p}\,t \ri\|_{L^{q}((0,\infty),\D
    t/t)}
\eeq
is finite.
For $1\le p<\infty$ and $1\le q_{1}\leq q_{2}\le\infty$, the continuous
embedding $L^{p,q_{1}}\subseteq L^{p,q_{2}}$ holds. Moreover, $L^{p,p}$
identifies with the Lebesgue space~$L^{p}$.
In the sequel, the notation $ \lesssim $ stands for the inequality $ \leq $ up to a multiplicative constant which is independent of the chosen function.

We use the following generalizations of H\"older and Young's inequality in
Lorentz spaces.
For $1\leq p_{1},p_{2}<\infty$, $1\le q_{1},q_{2}\le\infty$, H\"older's
inequality states that
\[
  \lf\|f_{1}f_{2} \ri\|_{L^{p,q}} \lesssim \lf\|f_{1} \ri\|_{L^{p_{1},q_{1}}} \lf\|f_{2} \ri\|_{L^{p_{2},q_{2}}},\qquad\frac{1}{p}=\frac{1}{p_{1}}+\frac{1}{p_{2}}\,,\quad\frac{1}{q}=\frac{1}{q_{1}}+\frac{1}{q_{2}}\,,
\]
whenever the r.h.s.~is finite.
Young's inequality states that, for $1<p,p_{1},p_{2}<\infty$, $1\le
q_{1},q_{2}\le\infty$,
\[
  \lf\|f_{1}*f_{2} \ri\|_{L^{p,q}}\lesssim \lf\|f_{1} \ri\|_{L^{p_{1},q_{1}}} \lf\|f_{2} \ri\|_{L^{p_{2},q_{2}}}\,,\qquad1+\frac{1}{p}=\frac{1}{p_{1}}+\frac{1}{p_{2}}\,,\quad\frac{1}{q}=\frac{1}{q_{1}}+\frac{1}{q_{2}}\,.
\]

Then, we have the following property of the Lorentz norms:

\begin{lem}
  \label{lem:control-of-FPsiPhi-by-Sobolev}
  If $\alpha_j\geq0$, $ j \in \{1,2\}$,
  $(\alpha_1+\alpha_2)/d=1/p\in(0,1]$ and $1\leq q\leq \infty$, then, for all
  $\psi_j \in \dot{H}^{\alpha_j}$,
  \begin{equation}
    \lf\|\mathcal{F}(\psi_1\psi_2) \ri\|_{L^{p,q}}\lesssim \lf\|\psi_1 \ri\|_{\dot{H}^{\alpha_1}} \lf\|\psi_2 \ri\|_{\dot{H}^{\alpha_2}}\,.
  \end{equation}
\end{lem}
\begin{proof}
  Using the Young inequality in Lorentz spaces with $\frac{1}{p_j} =
  \frac{\alpha_j}{d}+\frac{1}{2}$ yields
  \begin{equation}
    \lf\|\mathcal{F}(\psi_1\psi_2)\ri\|_{L^{p,q}}
    \lesssim \lf\|\psi_1\ri\|_{L^{p_1,2q}} \lf\|\psi_2\ri\|_{L^{p_2,2q}}\,
  \end{equation}
  we can then estimate the two terms on the right hand side in the same way,
  using the H\"older inequality in Lorentz spaces, and the continuity of the
  embedding $L^{2,2}\subseteq L^{2,2q}$,
  \begin{equation}
    \lf\|\psi_j\ri\|_{L^{p_1,2q}} \lesssim  \lf\| |k|^{-\alpha_j} \ri\|_{L^{d/\alpha_j,\infty}}  \lf\| |k|^{-\alpha_j} \psi_j \ri\|_{L^{2,2q}} 
    \lesssim   \lf\| |k|^{-\alpha_j} \psi_j \ri\|_{L^{2,2}} = \lf\|\psi_j \ri\|_{\dot{H}^{\alpha_j}}
  \end{equation}
  which achieves the proof.
\end{proof}

From now on we set the space dimension $d$ equal to 3. We state a simple Lorentz space estimate, which is useful both for the Nelson model and
the Pauli-Fierz model:
\begin{lem}
  \label{lem:Bound-psi-phi-chi-over-omega}
  If \cref{hyp:A-omega,hyp:A-chi} hold, and $\psi, \varphi$ are Lebesgue measurable functions such that $\psi, \varphi, \bar{\psi}\varphi  \in  \mathcal{S}'(\mathbb{R}^3)$, and $\mathcal{F}(\bar{\psi}\varphi) \in L^{6,2} $,
  then
  \[
    \lf\|\overline{\mathcal{F}}(\overline{\psi}\varphi)\frac{\chi}{\omega} \ri\|_{L^{2}}\lesssim \lf\|\mathcal{F}(\bar{\psi}\varphi) \ri\|_{L^{6,2}}\,.
  \]
\end{lem}

\begin{proof}
  H\"older's inequality in Lorentz spaces and \cref{hyp:A-chi} yield
  \begin{equation*}
    \lf\|\overline{\mathcal{F}}(\overline{\psi}\varphi)\tfrac{\chi}{\omega} \ri\|_{L^{2}}\leq \lf\|\overline{\mathcal{F}}(\overline{\psi}\varphi) \ri\|_{L^{6,2}} \lf\|\tfrac{\chi}{\omega} \ri\|_{L^{3,\infty}}
    \,,
  \end{equation*}
  which is the claimed result.
\end{proof}

The Lorentz spaces estimates we use for the Pauli-Fierz model are given
below.
\begin{lem}
  \label{lem:Bound-psi-phi-chi-over-omega-1} If \cref{hyp:A-omega} and \eqref{hyp:A-chi} hold, and $\psi, \varphi$ are Lebesgue measurable functions such $\psi, \varphi, \bar{\psi}\varphi  \in  \mathcal{S}'(\mathbb{R}^3)$, and $\mathcal{F}(\bar{\psi}\varphi) \in L^{3,2} $, then
  \[
    \lf\|\overline{\mathcal{F}}(\overline{\psi}\varphi)(k-k')\frac{\chi(k)}{\omega(k)}\frac{\chi(k')}{\omega(k')} \ri\|_{L_{k,k'}^{2}}
    \lesssim \lf\|\mathcal{F}(\bar{\psi}\varphi) \ri\|_{L^{3,2}}
  \]
  and $ \lf\|\mathcal{F}(\bar{\psi}\varphi) \ri\|_{L^{3,2}}$ can be controlled either
  by $ \lf\|\psi \ri\|_{\dot{H}^{1/2}} \lf\|\varphi \ri\|_{\dot{H}^{1/2}}$ or
  $ \lf\|\psi \ri\|_{L^{2}} \lf\|\varphi \ri\|_{\dot{H}^{1}}$.
\end{lem}

\begin{proof}
  By H\"older and Young's inequalities in Lorentz spaces
  \begin{multline*}
    \lf\| \lf(|\overline{\mathcal{F}}(\overline{\psi}\varphi)|^2*\tfrac{\chi^2}{\omega^2} \ri)\tfrac{\chi^2}{\omega^2} \ri\|_{L^1} 
    \lesssim \lf\|\tfrac{\chi^2}{\omega^2} \ri\|_{L^{3/2,\infty}} \lf\||\overline{\mathcal{F}}(\overline{\psi}\varphi)|^2*\tfrac{\chi^2}{\omega^2} \ri\|_{L^{3,1}} \\
    \lesssim \lf\|\tfrac{\chi}{\omega} \ri\|_{L^{3,\infty}}^2\Big\||\overline{\mathcal{F}}(\overline{\psi}\varphi)|^2\Big\|_{L^{3/2,1}} \lf\| \tfrac{\chi^2}{\omega^2} \ri\|_{L^{3/2,\infty}} 
    \lesssim \lf\|\tfrac{\chi}{\omega} \ri\|_{L^{3,\infty}}^4 \lf\| \overline{\mathcal{F}}(\overline{\psi}\varphi) \ri\|_{L^{3,2}} \,
  \end{multline*}
  and \cref{lem:control-of-FPsiPhi-by-Sobolev} yields the bounds on
  $ \lf\|\mathcal{F}(\bar{\psi}\varphi) \ri\|_{L^{3,2}}$.
\end{proof}

\subsection{Semiclassical Analysis}
\label{sec:semicl-analys}

We present here the main tool of semiclassical analysis ensuring the existence of at least one Wigner measure (see \cref{def:muNelson}) associated to the family of field's states $ \lf\{ \Psi_\eps \ri\}_{\eps \in (0,1)} $, as well as the convergence of expectations of suitable observables (monomials of creations and annihilation operators). In order to properly state the latter, let us introduce a class\footnote{We use the following convention:  if either $ \ell $ or $ m = 0 $, no factor of the corresponding kind is present in the symbol.} $ \mathscr{S}_{\ell,m} $, $ \ell, m \in \N $, of (cylindrical) classical symbols of the form
\beq
	\label{eq: cylindrical}
	\mathcal{S}(z) = \prod_{i = 1}^{\ell} \prod_{j = \ell+1}^{\ell+m} \braket{\omega^{1/2} z}{\omega^{-1/2} g_i}_{\mathfrak{h}} \braket{\omega^{-1/2} g_j}{\omega^{1/2} z}_{\mathfrak{h}},
\eeq
where $ g_i \in \mathfrak{h}_{\omega^{-1}} $ for any $ i \in \lf\{1, \ldots, \ell + m \ri\} $ and $ z \in \mathfrak{h}_{\omega} $. The Wick quantization of such symbols is simply given by the normal ordered monomial
\beq
	\mathrm{Op}_{\varepsilon}^{\mathrm{Wick}}\lf(\mathcal{S}(z) \ri) = \prod_{i = 1}^{\ell} \prod_{j = \ell+1}^{\ell+m}  a_{\varepsilon}^{\dagger}(g_i) a_{\varepsilon}(g_j).
\eeq

\begin{prop}[Existence of Wigner measures]
  \label{prop:semicl-analys-1}  \mbox{}	\\
 	Suppose that \cref{hyp:A-omega} holds. If there exists $ \delta > 0 $ such that the family of states $ \lf\{ \Psi_{\eps} \ri\}_{\eps \in (0,1)} \subset \mathscr{F}_{\eps} $ satisfies
 	\begin{equation}
    \lf| \meanlr{\Psi_{\eps}}{\lf(1 + \mathrm{d}\Gamma_{\varepsilon}(\omega) \ri)^{\delta}}{\Psi_{\varepsilon}}_{\mathscr{F}_{\varepsilon}}  \ri| \leq C\;,
  	\end{equation}
 uniformly in $ \eps $, then $ \mathscr{M}_{\omega} \lf(\Psi_{\eps} \ri) \neq \varnothing $ and
 	\beq
  		\int_{\mathfrak{h}_{\omega}} \lf( 1 + \lf\| \omega^{1/2} z \ri\|^2_{\mathfrak{h}} \ri)^{\eta} \diff \mu(z) < + \infty,	
  		\qquad		\forall \eta \leq \delta.
	\eeq
	Furthermore, if $  \Psi_{\varepsilon_n} \xrightarrow[n \to + \infty]{\omega
  \mathrm{-sc}} \mu  $ and $ \mathcal{S}(z) \in \mathscr{S}_{\ell,m} $ is a symbol of the form \eqref{eq: cylindrical} with $ \frac{\ell + m}{2} < \delta $, then
	\beq    	
		\braket{\Psi_{\eps_n}}{\mathrm{Op}_{\varepsilon_n}^{\mathrm{Wick}} \lf( \mathcal{S}(z) \ri) \Psi_{\eps_n}}_{\mathscr{F}_{\eps_n}} \xrightarrow[n \to \infty]{}  \int_{\mathfrak{h}_{\omega}}  \mathcal{S}(z) \, \diff \mu(z).
	\eeq 	 
\end{prop}
\begin{proof}
  The result concerning the existence of Wigner measures for $ \omega = 1 $ is already stated in \cite[Thm. 6.2]{ammari2008ahp}. The extension to a generic $\omega\geq 0$ and the result about the convergence of cylindrical symbols can be found in \cite[Thm.\ 3.3]{falconi2017ccm}.
\end{proof}

\begin{rem}[Nelson model]
  \label{rem:Nelson qc}
  \mbox{}	\\
 	Combining \cref{prop:semicl-analys-1} above with \cref{hyp:A-Nel}, we immediately see that, if the latter holds, then not only the set of Wigner measures is non-empty, but also for any $ \mu \in \mathscr{M}_{\omega}(\Psi_{\eps}) $ and any  subsequence such that $ \Psi_{\varepsilon_n} \xrightarrow[n \to + \infty]{\omega
  \mathrm{-sc}} \mu $,
 	 \beq
 	 	\label{eq:exp creation}
  		\braketr{\Psi_{\eps_n}}{a_{\varepsilon_n}^{\dagger}(g)  \Psi_{\eps_n}}_{\mathscr{F}_{\eps_n}}  
		\xrightarrow[n \to \infty]{}  \disp\int_{\mathfrak{h}^{\mathrm{Nel}}_{\omega}} \braketr{\omega^{1/2} z}{\omega^{-1/2} g}_{\mathfrak{h}^{\mathrm{Nel}}} \, \diff \mu(z),
	 \eeq
	for all $g \in \mathfrak{h}^{\mathrm{Nel}}_{\omega^{-1}}$. Obviously, the analogous result for the expectation of $ a_{\varepsilon_n}(g) $ holds true. 
\end{rem}

\begin{rem}[PF model]
  \label{rem:PF qc}
  \mbox{}	\\
 	Analogously to \cref{rem:Nelson qc}, the combination of \cref{prop:semicl-analys-1}  with \cref{hyp:A-PF} guarantees the non-emptiness of the set of Wigner measures $ \mathscr{M}_{\omega}(\Psi_{\eps}) $, as well as the convergence (over subsequences) of the expectation values of monomials of degree up to $ 2 $ of creation and annihilation operators, {\it i.e.}, in addition to the analogue of  \eqref{eq:exp creation}, one also has that 
 	\begin{eqnarray}
  		&& \braketr{\Psi_{\eps_n}}{\lf( a_{\varepsilon_n}^{\dagger}(g) \ri)^2 \Psi_{\eps_n}}_{\mathscr{F}_{\eps_n}}  
		\xrightarrow[n \to \infty]{}  \disp\int_{\mathfrak{h}_{\omega}^{\mathrm{PF}}} \lf( \braketr{\omega^{1/2} z}{\omega^{-1/2} g}_{\mathfrak{h}^{\mathrm{PF}}} \ri)^2 \, \diff \mu(z),	\label{eq:exp creation squared} \\
		&& \braketr{\Psi_{\eps_n}}{a_{\varepsilon_n}^{\dagger}(g) a_{\varepsilon_n}(g) \Psi_{\eps_n}}_{\mathscr{F}_{\eps_n}}  \xrightarrow[n \to \infty]{}  \int_{\mathfrak{h}_{\omega}^{\mathrm{PF}}} \lf| \braketr{\omega^{-1/2} g}{\omega^{1/2} z}_{\mathfrak{h}^{\mathrm{PF}}} \ri|^2 \, \diff \mu(z), \label{eq:exp creation annihilation}
	 \end{eqnarray}
	for all $g \in \mathfrak{h}^{\mathrm{PF}}_{\omega^{-1}}$. 
\end{rem}

\section{The Nelson model: Proof of \cref{thm:Quasiclassical-limit-Nelson}}\label{sec:Nelson}

We present first the method on the Nelson model, where most ideas can already
be understood. Throughout this section, the spaces $L^{p}$, $L^{p,q}$ and $H^{s}$ have their
variables in $\mathbb{R}^{3}$.

\cref{lem:Bound-psi-phi-chi-over-omega} allows us to make sense of the
potential $V_\varepsilon$ (recall its definition in \eqref{eq:def-Vepsilon}) once the fields are traced out, and also yields a
useful estimate:
\begin{prop}[Estimate of $ V_\eps $]
	\label{prop:estimate-Vepsilon}
	\mbox{}	\\
	Suppose that \cref{hyp:A-omega,hyp:A-chi,hyp:A-Nel}
  hold. If $\mathcal{F}(\bar{\psi}\varphi) \in L^{6,2}$, then, 
  \begin{equation}\label{eq:bound_Veps}
    \lf| \braket{\psi}{V_{\varepsilon}\varphi}_{L_{x}^{2}} \ri| \lesssim \lf\|\mathcal{F}(\bar{\psi}\varphi) \ri\|_{L^{6,2}} \meanlr{\Psi_{\varepsilon}}{\D\Gamma_{\varepsilon}(\omega)}{\Psi_{\varepsilon}}^{1/2} \lf\|\Psi_{\varepsilon} \ri\|_{\mathscr{F}_{\eps}}\,.
  \end{equation}
  Moreover, 
  $V_{\varepsilon}\:|\vecp |^{-1/2} \in \mathscr{B}(L^{2}) $ uniformly in $\varepsilon$.
\end{prop}
\begin{proof}
  Taking the modulus of both sides of the identity
  \begin{equation*}
    \braketr{\psi\otimes\Psi_{\varepsilon}}{a_{\varepsilon}^{*}\lf(e^{ix\cdot k}\tfrac{\chi}{\sqrt{\omega}}\ri)\,\varphi\otimes\Psi_{\varepsilon} }
    =\int \overline\psi(x)\varphi(x)e^{ix\cdot k}\frac{\chi(k)}{\sqrt{\omega(k)}} \braketr{ \Psi_\varepsilon}{a^*_\varepsilon(k) \Psi_\varepsilon}_{\mathscr{F}_{\eps}}\D x\D k
  \end{equation*}
  yields
  \begin{equation*}
   	\lf| \braketr{\psi\otimes\Psi_{\varepsilon}}{a_{\varepsilon}^{*}\lf(e^{ix\cdot k}\tfrac{\chi}{\sqrt{\omega}}\ri)\,\varphi\otimes\Psi_{\varepsilon} } \ri|
    \leq \lf\|\overline{\mathcal{F}}(\overline{\psi}\varphi)\tfrac{\chi}{\omega} \ri\|_{L^{2}_k} \lf\| \sqrt{\omega(\cdot)}\lf\| a_\varepsilon(\cdot)\Psi_\varepsilon \ri\|_{\mathscr{F}_{\eps}} \ri\|_{L^2_k} \lf\|\Psi_\varepsilon \ri\|_{\mathscr{F}_{\eps}}.
  \end{equation*}
  We have that
  $ \lf\| \sqrt{\omega(\cdot)} \lf\|a_\varepsilon(\cdot)\Psi_\varepsilon \ri\|_{\mathscr{F}_{\eps}} \ri\|_{L^2_k}= \meanlr{\Psi_{\varepsilon}}{\D\Gamma_{\varepsilon}(\omega)}{\Psi_{\varepsilon}}^{1/2} $
  and the term  $ \lf\|\overline{\mathcal{F}}(\overline{\psi}\varphi)\frac{\chi}{\omega}\ri\|_{L^{2}_k}$
 is bounded by \cref{lem:Bound-psi-phi-chi-over-omega}. This shows
  \eqref{eq:bound_Veps}. The facts that the operator 
  $V_{\varepsilon}\:|\vecp |^{-1/2}$ is bounded on $L^{2}$
  uniformly in $\varepsilon$  follows from
  \cref{lem:control-of-FPsiPhi-by-Sobolev}.
\end{proof}

We can then make sense out of the effective potential $V_\mu$ (recall its definition in  \eqref{eq:def-Vmu}) obtained as the quasi-classical limit of
$V_\varepsilon$.

\begin{prop}[Estimates of $ V_{\mu} $]
  \label{prop:unif2} 
  \mbox{}	\\ 
  Suppose that \cref{hyp:A-omega,hyp:A-chi,hyp:A-Nel}
  hold. Let $\mu \in \mathscr{M}_{\omega}(\Psi_\varepsilon)$. Then, 
  \beq
     \lf| \meanlr{\psi}{V_{\mu}}{\varphi} \ri| \lesssim \lf\| \mathcal{F}(\bar{\psi}\varphi) \ri\|_{L^{6,2}} \lf( \int_{\mathfrak{h}_{\omega}} \lf\| \omega^{1/2} z \ri\|_{L^{2}}^{2}
    \, \D \mu(z) \ri)^{1/2}\,.
  \eeq
  Moreover, 
  $V_{\mu}\:|\vecp |^{-1/2} \in \mathscr{B}(L^{2})$.
\end{prop}

\begin{proof}
  The Cauchy-Schwarz inequality followed by
  \cref{lem:Bound-psi-phi-chi-over-omega} yield
  \bmln{
 	\lf| \meanlr{\psi}{V_{\mu}}{\varphi} \ri| \lesssim \lf\|\mathcal{F}(\psi\overline{\varphi})\chi/\omega \ri\|_{L^{2}}\lf( \int_{\mathfrak{h}_{\omega}} \lf\| \omega^{1/2} z \ri\|_{L^{2}}^{2}
    \, \D \mu(z) \ri)^{1/2} \\
    \lesssim \lf\|\mathcal{F}(\bar{\psi}\varphi) \ri\|_{L^{6,2}} \lf\|\tfrac{\chi}{\omega} \ri\|_{L^{3,\infty}} \lf( \int_{\mathfrak{h}_{\omega}} \lf\| \omega^{1/2} z \ri\|_{L^{2}}^{2}
    \, \D \mu(z) \ri)^{1/2}\, .
  }
  As in the previous proof, \cref{lem:control-of-FPsiPhi-by-Sobolev}
  shows that the operator 
  $V_{\mu}\:|\vecp |^{-1/2}$ 
  is bounded on
  $L^{2}$.
\end{proof}

We are now ready to prove a convergence result of $V_\varepsilon$ towards
$V_\mu$:
\begin{prop}[Convergence of $ V_{\eps} $]
  \label{prop:Convergence-V}
  \mbox{}	\\
  Suppose that \cref{hyp:A-omega,hyp:A-chi,hyp:A-Nel} hold. Let $\mu \in \mathscr{M}_{\omega}(\Psi_\varepsilon)$ and $ \lf\{ \eps_n \ri\}_{n \in \N} $ be such that $  \Psi_{\varepsilon_n} \xrightarrow[n \to +
  \infty]{\omega \mathrm{-sc}} \mu $. If
  $\mathcal{F}(\bar{\psi}\varphi) \in L^{6,2}$, then
  \begin{equation}
	    \int_{\R^3} (V_{\varepsilon_n}-V_{\mu}) \, \bar{\psi}\varphi  \xrightarrow[n \to +\infty]{}0.\label{eq:conv-V}
  \end{equation}
  Moreover, the family of operators
  $(V_{\varepsilon_n}-V_{\mu})\:|\vecp |^{-1/2} $ is bounded on $ L^{2}$ uniformly in $ n $ and converges weakly to
  $0$.
\end{prop}
\begin{proof}
  If $\mathcal{F}(\bar{\psi}\varphi) \in L^{6,2}$, we have
  \bmln{
    \meanlr{\psi}{V_{\varepsilon_n}}{\varphi}_{L_{x}^{2}}
    = \meanlr{\psi\otimes\Psi_{\varepsilon_n}}{\mathrm{Re} \, a_{\varepsilon_n}^{\dagger}\lf(e^{ix\cdot k}\tfrac{\chi}{\sqrt{\omega}}\ri)}{\varphi\otimes\Psi_{\varepsilon_n}} \\
    = \braketr{\Psi_{\varepsilon_n}}{a_{\varepsilon_n}^{\dagger} \lf( \braketr{\psi}{e^{ix\cdot k}\varphi}_{L_{x}^{2}} \tfrac{\chi}{\sqrt{\omega}} \ri) \Psi_{\varepsilon_n}}_{\mathscr{F}_{\eps_n}}
    + \braketr{\Psi_{\varepsilon_n}}{a_{\varepsilon_n} \lf( \braketr{\varphi}{e^{ix\cdot k}\psi}_{L_{x}^{2}} \tfrac{\chi}{\sqrt{\omega}}\ri) \Psi_{\varepsilon_n}}_{\mathscr{F}_{\eps_n}}\\
    \xrightarrow[n \to +\infty]{} \int_{\mathfrak{h}_{\omega}}  \meanlr{\psi}{\mathrm{Re}\,\lf( \braketr{\omega^{1/2} z}{e^{ix\cdot k} \tfrac{\chi}{\omega}}_{L_{k}^{2}} \ri)}{ \varphi}_{L_{x}^{2}}\, \D \mu(z)\,,
	}
	 where we used that $\mu \in \mathscr{M}_{\omega}(\Psi_\varepsilon)$
  (see \cref{def:muNelson}) and that $\overline{\mathcal{F}}(\overline{\psi}\varphi)\chi/\omega \in L^2$ by \cref{lem:Bound-psi-phi-chi-over-omega}.

  The uniform boundedness in $n$ of the family of operators
  $(V_{\varepsilon_n}-V_{\mu})\:|\vecp |^{-1/2}$
  in $L^{2}$ follows from \cref{prop:estimate-Vepsilon,prop:unif2}.
  The fact that $(V_{\varepsilon_n}-V_{\mu})\:|\vecp |^{-1/2}$ converges weakly
  to $0$ follows from \cref{lem:control-of-FPsiPhi-by-Sobolev} with
  $\alpha_1=0$, $\alpha_2=1/2$ and $d=3$: it shows that if $\varphi \in \dot{H}^{1/2}$ and $\psi\in L^2$, then
  $\mathcal{F}(\bar{\psi}\varphi) \in L^{6,2}$ and we can then
  conclude thanks to \eqref{eq:conv-V}.
\end{proof}

\begin{rem}
	\mbox{}	\\
  The previous proof also shows that
  $ \mathcal{F}(V_{\varepsilon_n}) \xrightarrow[n \to+ \infty]{} \mathcal{F}( V_{\mu}) $ weakly
  in~$L^{6/5,2}$.
\end{rem}

We are now ready to prove the main result of this section on the Nelson
model.

\begin{proof}[Proof of  \cref{thm:Quasiclassical-limit-Nelson}]
  Let us fix $\mu \in \mathscr{M}_{\omega}(\Psi_\varepsilon)$, and the
  sequence $\{\varepsilon_n\}_{n\in \mathbb{N}}$ such that
  $\Psi_{\varepsilon_n}\to \mu$.
  
  To justify that $H^{\mathrm{Nel}}_{\varepsilon_n}=-\Delta+U+V_{\varepsilon_n}$ and
  $\Hnelm=-\Delta+U+V_{\mu}$ define symmetric closed quadratic forms with form
  domain $\mathcal{Q}:=\mathcal{Q}(H^{\mathrm{Nel}}_{\varepsilon_n})=\mathcal{Q}(\Hnelm)=H^{1}(\R^3) \cap L^{2}(\R^3,
  U_{+}\, \D x)$, it suffices to argue as follows: first, since $U_+$ is
  non-negative and belongs to $L^1_{\mathrm{loc}}$, $-\Delta+U_+$ identifies with
  a self-adjoint operator with form domain $\mathcal{Q}$. Next, \cref{hyp:A-U} together
  with \cref{prop:estimate-Vepsilon} show that $U_-+V_{\varepsilon_n}$ is
  relatively form bounded w.r.t. $-\Delta$ (and hence w.r.t.
  $-\Delta+U_+$) with a relative bound $<1$. The KLMN Theorem (see
  e.g. \cite[Theorem X.17]{reed1975II}) then shows that
  $H^{\mathrm{Nel}}_{\varepsilon_n}$ identifies with a self-adjoint operator with form
  domain $\mathcal{Q}$. The same holds for $\Hnelm$, using \cref{prop:unif2}
  instead of \cref{prop:estimate-Vepsilon}.
  
 Now the goal is to prove $\Gamma$-convergence of
  the quadratic form $ \meanlr{\varphi}{H^{\mathrm{Nel}}_n}{\varphi}$ to
  $ \meanlr{\varphi}{H_{\mu}}{\varphi}$ as $n\to \infty$, in the common
  domain of definition $\mathcal{Q}$.

  Let us start with the $\Gamma\text{-}\limsup$. For any $\varphi\in
  \mathcal{Q}$, we have to construct a sequence $\{\varphi_{n}\}_{n\in
    \mathbb{N}}\subset \mathcal{Q}$ such that
  \begin{equation}
    \label{eq:gammalimsup}
    \,\varphi_n \xrightarrow[n\to +\infty]{\mathcal{Q}} \varphi\;,\qquad \limsup_{n\to +\infty}\: \meanlr{\varphi_n}{H^{\mathrm{Nel}}_n}{\varphi_n} \leq \meanlr{\varphi}{H_{\mu}}{\varphi}\;.
  \end{equation}
  In view of \cref{prop:Convergence-V}, it is sufficient to choose
  $\varphi_n\equiv \varphi$ the constant sequence, to get
  \begin{equation*}
    \lim_{n\to +\infty}\: \meanlr{\varphi_n}{H^{\mathrm{Nel}}_n}{\varphi_n} =  \meanlr{\varphi}{H_{\mu}}{\varphi}\;.
  \end{equation*}

For the $\Gamma\text{-}\liminf$, we apply  \cref{prop:2}: the fact that the assumptions of \cref{prop:2} are satisfied follows from \cref{prop:estimate-Vepsilon,prop:unif2,prop:Convergence-V} (note in particular that \eqref{eq:unif-bounds} holds with $\delta=1/4$). This concludes the proof.

\end{proof}

\section{The Pauli-Fierz  Model: Proof of  \cref{thm:Quasi-classical-limit-Pauli-Fierz}}\label{sec:Pauli-Fierz}

In this section, we set $L^{p} : = L^p(\R^3 \otimes \C^2) $, $L^{p,q} : = L^{p,q}(\R^3 \otimes \C^2) $ and $H^{s} : = H^s(\R^3 \otimes \C^2)$.
%
%
Recall the definition of the vector potential $\vec{A}_\varepsilon$ in \eqref{eq:def-Aepsilon} with form domain $H^{1}\cap
L^{2}(U_{+}\, \D x)$. With the
same proof as \cref{prop:estimate-Vepsilon} for linearly coupled
models, one gets:

\begin{prop}[Estimate of $ \vec{A}_{\eps} $]\label{prop:estimate-Aepsilon}
	\mbox{}	\\
	Suppose that \cref{hyp:A-omega,hyp:A-chi,hyp:A-PF}
  hold. If $\mathcal{F}(\bar{\psi}\varphi) \in L^{6,2}$, then, 
  \begin{equation}\label{eq:bound_Veps}
    \lf| \braket{\psi}{\vec{A}_{\varepsilon}\varphi}_{L_{x}^{2}} \ri| \lesssim \lf\|\mathcal{F}(\bar{\psi}\varphi) \ri\|_{L^{6,2}} \meanlr{\Psi_{\varepsilon}}{\D\Gamma_{\varepsilon}(\omega)}{\Psi_{\varepsilon}}^{1/2} \lf\|\Psi_{\varepsilon} \ri\|_{\mathscr{F}_{\eps}}\,.
  \end{equation}
  Moreover, 
  $\vec{A}_{\varepsilon}\:|\vecp |^{-1/2} \in \mathscr{B}(L^{2}) $ uniformly in $\varepsilon$.	
\end{prop}

Next we decompose the potential $W_\varepsilon$ defined in
\eqref{eq:def-Wepsilon} as
\[
  W_\varepsilon=W_{a^{*}a^*,\varepsilon}+W_{aa,\varepsilon}+2W_{a^{*}a,\varepsilon}
\]
with
\begin{equation}\label{eq:Waa-star}
  W_{a^{*}a,\varepsilon}(x) = \lf\| a_{\varepsilon} \lf(e^{ix\cdot k}\tfrac{\chi}{\sqrt{\omega}} \ri)\, \Psi_{\varepsilon}\ri\|^2_{\mathscr{F}_{\eps}},
\end{equation}
\begin{equation}\label{eq:Waa}
  W_{aa,\varepsilon}(x) = \braketr{\Psi_{\varepsilon}}{a_{\varepsilon}^2 \lf(e^{ix\cdot k}\tfrac{\chi}{\sqrt{\omega}} \ri) \, \Psi_{\varepsilon}}_{\mathscr{F}_{\eps}},
\end{equation}
and likewise for $W_{a^*a^*,\varepsilon}(x) =  \overline{W_{aa,\varepsilon}(x)}  $.

\begin{prop}[Estimate of $ W_{\eps} $]
	\label{prop:estimate-Wepsilon}
	\mbox{}	\\
  Suppose that \cref{hyp:A-omega,hyp:A-chi,hyp:A-PF}
  hold.  If $\mathcal{F}(\bar{\psi}\varphi) \in L^{3,2}$, then
  	\beq
    	\lf| \meanlr{\psi}{W_{a^{*}a,\varepsilon}}{\varphi}_{L_{x}^{2}} \ri| \lesssim \lf\|\mathcal{F}(\bar{\psi}\varphi) \ri\|_{L^{3,\infty}} \meanlr{\Psi_{\varepsilon}}{\D\Gamma_{\varepsilon}(\omega)}{\Psi_{\varepsilon}}\,
    ,
	\eeq  
	and
  	\beq
    		\lf| \braketr{\psi}{W_{aa,\varepsilon}\varphi}_{L_{x}^{2}} \ri| \lesssim \lf\|\mathcal{F}(\bar{\psi}\varphi) \ri\|_{L^{3,2}}
   \lf\|\Psi_\varepsilon\ri\|_{\mathscr{F}_{\eps}} \meanlr{\Psi_{\varepsilon}}{\D\Gamma_{\varepsilon}^{(2)}(\omega\otimes\omega)}{\Psi_{\varepsilon}}^{1/2}\,
    .
  	\eeq
  	As a consequence,
	\beq    
		\lf|   \meanlr{\psi}{W_{\varepsilon}}{\varphi}_{L_{x}^{2}} \ri|
    \lesssim\lf\|\mathcal{F}(\bar{\psi}\varphi) \ri\|_{L^{3,2}}\, ,
	\eeq  
	and $W_{\varepsilon}\:|\vecp |^{-1} \in \mathscr{B}(L^{2})$
  uniformly in $\varepsilon$.
\end{prop}
%
%
%
\begin{proof}
  By the definition of $W_{a^*a,\varepsilon}$\,, we have
  \begin{align*}
     \meanlr{\psi}{W_{a^{*}a,\varepsilon}}{\varphi}_{L_{x}^{2}} & = \braketr{a_{\varepsilon} \lf(e^{ix\cdot k}\tfrac{\chi}{\sqrt{\omega}}\ri)\,\psi\otimes\Psi_{\varepsilon}}{a_{\varepsilon} \lf(e^{ix\cdot k}\tfrac{\chi}{\sqrt{\omega}} \ri)\,\varphi\otimes\Psi_{\varepsilon}} \\
    &=\iiint \bar{\psi}(x) \varphi(x) e^{-ix\cdot (k-k')} \tfrac{\chi(k)}{\sqrt{\omega(k)}}\tfrac{\chi(k')}{\sqrt{\omega(k')}} \braketr{ a_\varepsilon(k)\Psi_\varepsilon}{a_\varepsilon(k')\Psi_\varepsilon}_{\mathscr{F}_{\eps}}\,\D x\D k\D k' \\
    &=\iint \mathcal{F}(\bar{\psi} \varphi)(k-k') \tfrac{\chi(k)}{\sqrt{\omega(k)}}\tfrac{\chi(k')}{\sqrt{\omega(k')}} \braketr{ a_\varepsilon(k)\Psi_\varepsilon}{a_\varepsilon(k')\Psi_\varepsilon}_{\mathscr{F}_{\eps}}\,\D k\D k' ,
  \end{align*}
  and hence, taking the modulus of the previous equalities provides the bound
  \begin{align*}
    \lf| \meanlr{\psi}{W_{a^{*}a,\varepsilon}}{\varphi}_{L_{x}^{2}} \ri|
    &\le \iint \lf| \mathcal{F}(\bar{\psi} \varphi)(k-k') \ri| \tfrac{\chi(k)}{\sqrt{\omega(k)}}\tfrac{\chi(k')}{\sqrt{\omega(k')}} \lf\| a_\varepsilon(k)\Psi_\varepsilon \ri\|_{\mathscr{F}_{\eps}} \lf\| a_\varepsilon(k')\Psi_\varepsilon\ri\|_{\mathscr{F}_{\eps}} \,\D k\D k' \\
    &=\int \tfrac{\chi(k)}{\sqrt{\omega(k)}} \lf\| a_\varepsilon(k)\Psi_\varepsilon \ri\|_{\mathscr{F}_{\eps}} \lf[ \lf| \mathcal{F}(\bar{\psi} \varphi) \ri| * \tfrac{\chi(\cdot)}{\sqrt{\omega(\cdot)}} \lf\| a_\varepsilon(\cdot)\Psi_\varepsilon \ri\|_{\mathscr{F}_{\eps}} \ri](k)\D k\,.
  \end{align*}
  H\"older, Young and again H\"older's inequalities in Lorentz spaces yield
  \begin{align*}
    \lf| \meanlr{\psi}{W_{a^{*}a,\varepsilon}}{\varphi}_{L_{x}^{2}} \ri| & \lesssim \lf\| \tfrac{\chi(k)}{\sqrt{\omega(k)}} \lf\| a_\varepsilon(k)\Psi_\varepsilon \ri\|_{\mathscr{F}_{\eps}} \ri\|_{L^{6/5,2}_k} \lf\| \lf| \mathcal{F}(\bar{\psi} \varphi) \ri| * \tfrac{\chi(\cdot)}{\sqrt{\omega(\cdot)}} \lf\| a_\varepsilon(\cdot)\Psi_\varepsilon \ri\|_{\mathscr{F}_{\eps}} \ri\|_{L^{6,2}} \\
    &\lesssim \lf\| \tfrac{\chi(k)}{\sqrt{\omega(k)}} \lf\| a_\varepsilon(k)\Psi_\varepsilon \ri\|_{\mathscr{F}_{\eps}} \ri\|_{L^{6/5,2}_k}^2 \lf\| \mathcal{F}(\bar{\psi} \varphi) \ri\|_{L^{3,\infty}} \\
    &\lesssim \lf\| \tfrac{\chi}{\omega}  \ri\|_{L^{3,\infty}}^2 \lf\| \lf\| \omega(\cdot)^{1/2} a_\varepsilon(\cdot)\Psi_\varepsilon \ri\|_{\mathscr{F}_{\eps}} \ri\|_{L^2}^2 \lf\| \mathcal{F}(\bar{\psi} \varphi) \ri\|_{L^{3,\infty}} \\
    &= \lf\|\tfrac{\chi}{\omega}  \ri\|_{L^{3,\infty}}^2  \lf\| \mathcal{F}(\bar{\psi} \varphi) \ri\|_{L^{3,\infty}} \meanlr{\Psi_\varepsilon}{\D \Gamma_\varepsilon(\omega)}{\Psi_\varepsilon}.
  \end{align*}

  Next, by the definition of $W_{aa,\varepsilon}$\,,
  \begin{align*}
    & \braketr{\psi}{W_{aa,\varepsilon}\varphi}_{L_{x}^{2}}= \braketr{\psi\otimes\Psi_{\varepsilon}}{a_{\varepsilon}^2 \lf(e^{ix\cdot k}\tfrac{\chi}{\sqrt{\omega}} \ri)\, \varphi\otimes\Psi_{\varepsilon}}\\
    & =\iint \mathcal{F}(\bar{\psi}\varphi)(-k-k')\tfrac{\chi(k)}{\omega(k)}\tfrac{\chi(k')}{\omega(k')} \sqrt{\omega(k)\omega(k')} \braketr{\Psi_\varepsilon}{a_\varepsilon(k)a_\varepsilon(k')\Psi_\varepsilon}_{\mathscr{F}_{\eps}}\,\D k\D k'\,.
  \end{align*}
  The Cauchy-Schwarz inequality in $L^{2}_{k,k'}$ along with
   \cref{lem:Bound-psi-phi-chi-over-omega-1} yield
  \begin{align*}
    & \lf| \braketr{\psi}{W_{aa,\varepsilon}\varphi}_{L_{x}^{2}} \ri|\\
    & \leq \lf\|\mathcal{F}(\bar{\psi}\varphi)(-k-k')\tfrac{\chi(k)}{\omega(k)}\tfrac{\chi(k')}{\omega(k')} \ri\|_{L_{k,k'}^{2}} \lf\| \sqrt{\omega(k)\omega(k')} \braketr{\Psi_\varepsilon}{a_\varepsilon(k)a_\varepsilon(k')\Psi_\varepsilon}_{\mathscr{F}_{\eps}}  \ri\|_{L_{k,k'}^{2}} \\
    & \lesssim \lf \|\tfrac{\chi}{\omega} \ri\|_{L^{3,\infty}}^2 \lf\| \mathcal{F}(\bar{\psi}\varphi) \ri\|_{L^{3,2}} \lf\|\Psi_\varepsilon \ri\|_{\mathscr{F}_{\eps}} \;\lf\| \sqrt{\omega(k)\omega(k')} \lf \|a_\varepsilon(k)a_\varepsilon(k')\Psi_\varepsilon \ri\|_{\mathscr{F}_{\eps}} \ri\|_{L_{k,k'}^{2}} \\
    & = \lf\|\tfrac{\chi}{\omega} \ri\|_{L^{3,\infty}}^2 \lf\| \mathcal{F}(\bar{\psi}\varphi) \ri\|_{L^{3,2}} \lf\|\Psi_\varepsilon\ri\|_{\mathscr{F}_{\eps}} \meanlr{\Psi_{\varepsilon}}{\D \Gamma_{\varepsilon}^{(2)}(\omega\otimes\omega)}{\Psi_{\varepsilon}}^{1/2}
  \end{align*}
  where we recall that $\D \Gamma_{\varepsilon}^{(2)}(\omega\otimes\omega)=\D
  \Gamma_{\varepsilon}(\omega)^{2}-\varepsilon \D
  \Gamma_{\varepsilon}(\omega^{2})$~.

  We deduce the uniform bound on the norm of $W_{\varepsilon}\:|\vecp |^{-1}$ from
   \cref{lem:control-of-FPsiPhi-by-Sobolev} with $\alpha_1=0$,
  $\alpha_2=1$ and $d=3$, since, if $\varphi \in \dot{H}^1$ and
  $\psi \in L^2$, then $\mathcal{F}(\bar{\psi}\varphi) \in L^{3,2}$.
\end{proof}

In our proof of the $\Gamma$-convergence below, we will need a
related estimate on
\begin{equation*}
  \vec{B}_\varepsilon := \vec{\nabla}\wedge\vec{A}_\varepsilon.
\end{equation*}
In this respect, it turns out that \cref{hyp:A-chi} is not
sufficient for our purpose. The next result holds under the slightly stronger
 \cref{hyp:A'-chi}.

\begin{prop}[Estimate of $ \mathbf{B}_{\eps} $]
	\label{prop:est Beps}
  \mbox{}	\\
  Suppose that \cref{hyp:A-omega,hyp:A'-chi,hyp:A-PF}
  hold.  If $\bar{\psi}\varphi \in L^{2}$, then
  \[
    \lf| \meanlr{\psi}{\vec{B}_{\varepsilon}}{\varphi} \ri| \lesssim \lf\|
    \bar{\psi}\varphi \ri\|_{L^{2}}.
  \]
  As a consequence, $ |\vecp |^{-3/4} \vec{B}_{\varepsilon}\:|\vecp
  |^{-3/4}  \in \mathscr{B}(L^{2})$ uniformly in
  $\varepsilon$.
\end{prop}
\begin{proof}
  Note that
  \begin{equation*}
    \meanlr{\psi}{\vec{B}_{\varepsilon}(x)}{\varphi} = \meanlr{\psi\otimes\Psi_{\varepsilon}}{\vec{\mathbb{B}}_{\varepsilon}(x)}{\varphi\otimes\Psi_{\varepsilon}},
  \end{equation*}
  where
  \[
    \vec{\mathbb{B}}_\varepsilon(x)=\sqrt{\varepsilon}\sum_{\lambda=1}^{2}\int
    \tfrac{\chi(k)}{\sqrt{\omega(k)}} k\wedge \vec{e}_{\lambda}(k) \lf(e^{ik\cdot
      x}a_{\lambda}(k)+e^{-ik\cdot x}a_{\lambda}^{\dagger}(k) \ri )\, \D k\,.
  \]
  Proceeding in the same way as for $\vec{A}_\varepsilon$ or $V_\varepsilon$,
  we can thus estimate
  \begin{equation*}
    \lf| \meanlr{\psi}{\vec{B}_{\varepsilon}(x)}{\varphi}  \ri| \lesssim \lf\| \tfrac {|k|\chi}{\omega} \mathcal{F}(\bar{\psi}\varphi) \ri\|_{L^{2}} \lesssim \lf\| \bar{\psi}\varphi \ri\|_{L^{2}},
  \end{equation*}
  since $|k|\chi/\omega \in L^\infty$ thanks to \cref{hyp:A'-chi}.

  The uniform boundedness of $ |\vecp |^{-3/4}
  \vec{B}_{\varepsilon}\:|\vecp |^{-3/4} $ then follows from
  \cref{lem:control-of-FPsiPhi-by-Sobolev} with $\alpha_1=\alpha_2=3/4$
  and $d=3$.
\end{proof}

Recall that $ \vec{A}_\mu$ has been defined in \eqref{eq:def-Amu}. Similarly as for
$\vec{B}_\varepsilon$, we set
\begin{equation*}
  \vec{B}_\mu := \vec{\nabla}\wedge\vec{A}_\mu.
\end{equation*}

\begin{prop}[Estimate of $ \mathbf{B}_{\mu} $]
  \label{prop:unif2-1} 
  \mbox{}	\\
  Suppose that \cref{hyp:A-omega,hyp:A-chi}
  hold. Let $\mu \in \mathscr{M}_{\omega}(\Psi_\varepsilon)$.  If $\mathcal{F}(\bar{\psi}\varphi) \in L^{6,2}$, then
 	\beq
     		\lf| \meanlr{\psi}{\vec{A}_{\mu}}{\varphi} \ri| \lesssim \lf\|\mathcal{F}(\bar{\psi}\varphi) \ri\|_{L^{6,2}}
    \lf( \int_{\mathfrak{h}_{\omega}} \lf \|\omega^{1/2}z \ri\|_{L^{2}}^{2}\,\D \mu(z) \ri)^{1/2}\,
  	\eeq
  and $\vec{A}_{\mu}\:|\vecp |^{-1/2} \in \mathscr{B}(L^2) $.  Likewise, if $\mathcal{F}(\bar{\psi}\varphi) \in L^{3,2}$, then
  	\beq
    		\lf| \meanlr{\psi}{W_{\mu}}{\varphi} \ri| \lesssim \lf\|\mathcal{F}(\bar{\psi}\varphi) \ri\|_{L^{3,2}}
	\eeq
	 and $W_{\mu}\:|\vecp |^{-1} \in \mathscr{B}(L^{2}) $.
  
  Moreover, if \cref{hyp:A-chi} is replaced by the stronger \cref{hyp:A'-chi} and if $\bar{\psi}\varphi \in L^{2}$, then
	\beq    
		\lf| \meanlr{\psi}{\vec{B}_{\mu}}{\varphi} \ri| \lesssim \lf\|\bar{\psi}\varphi \ri\|_{L^{2}}	
	\eeq  
	and $ |\vecp |^{-3/4} \vec{B}_{\mu}\:|\vecp |^{-3/4}  \in \mathscr{B}(L^{2}) $.
\end{prop}

\begin{proof}
  To prove the first estimate, it suffices to write
	\bmln{	    
		\lf| \meanlr{\psi}{\vec{A}_{\mu}}{\varphi} \ri| \lesssim \lf\|\mathcal{F}(\psi\overline{\varphi})\chi/\omega \ri\|_{L^{2}}\lf( \int_{\mathfrak{h}_{\omega}} \lf \|\omega^{1/2}z \ri\|_{L^{2}}^{2}\,\D \mu(z) \ri)^{1/2} \\
		\lesssim \lf\|\mathcal{F}(\bar{\psi}\varphi) \ri\|_{L^{6,2}} \lf\| \tfrac{\chi}{\omega} \ri\|_{L^{3,\infty}} \lf( \int_{\mathfrak{h}_{\omega}} \lf \|\omega^{1/2}z \ri\|_{L^{2}}^{2}\,\D \mu(z) \ri)^{1/2}
	}  
	where we used \cref{lem:Bound-psi-phi-chi-over-omega-1} as before. The
  statements concerning $W_\mu$ and $\vec{B}_\mu$ can be proven in the same
  way.
\end{proof}

\begin{prop}[Convergence of $ \vec{A}_{\eps} $, $W_{\eps} $, $ \vec{B}_{\eps} $]
  \label{prop:Convergence-V-1}
  \mbox{}	\\
  Suppose that \cref{hyp:A-omega,hyp:A-chi,hyp:A-PF} hold. Let $\mu \in \mathscr{M}_{\omega}(\Psi_\varepsilon)$ and $ \lf\{ \eps_n \ri\}_{n \in \N} $ be such that $  \Psi_{\varepsilon_n} \xrightarrow[n \to +
  \infty]{\omega \mathrm{-sc}} \mu $.
  If $\mathcal{F}(\bar{\psi}\varphi) \in L^{6,2}$, then
  \begin{equation}
    	\int_{\R^3} \lf(\vec{A}_{\varepsilon_n}-\vec{A}_{\mu} \ri) \, \bar{\psi}\varphi  \xrightarrow[n \to + \infty]{}0\label{eq:conv-A}
  \end{equation}
  and the family of operators $(\vec{A}_{\varepsilon_n}-\vec{A}_{\mu})\:|\vecp
  |^{-1/2}$ is bounded on $L^2$ uniformly in $n $ and converges weakly to $0$.  Likewise, if $\mathcal{F}(\bar{\psi}\varphi) \in L^{3,2}$, then
  \beq
    \int_{\R^3} \lf(W_{\varepsilon_n}-W_{\mu}\ri) \, \bar{\psi}\varphi \xrightarrow[n \to + \infty]{}0
  \eeq
  and the family of operators $(W_{\varepsilon_n}-W_{\mu})\:|\vecp |^{-1}$ is bounded on $L^2$ uniformly in $n$ and converges weakly to $0$.
  
Moreover, if \cref{hyp:A-chi} is replaced by the stronger \cref{hyp:A'-chi} and if $\bar \psi \varphi \in L^{2}$, then
  \begin{equation}
    \int_{\R^3} \lf(\vec{B}_{\varepsilon_n}-\vec{B}_{\mu}\ri) \, \bar{\psi}\varphi \xrightarrow[n \to +\infty]{}0,
  \end{equation}
  and the family of operators $ |\vecp
  |^{-3/4}(\vec{B}_{\varepsilon_n}-\vec{B}_{\mu})\:|\vecp |^{-3/4}$  is bounded on $L^2$ uniformly in $n$ and converges weakly to $0$.
\end{prop}

\begin{proof}The proof of the convergence of $\vec{A}_{\varepsilon_n}$ is the
  same as the proof of \cref{prop:Convergence-V}.

  Likewise, using that $\mu \in \mathscr{M}_{\omega}(\Psi_\varepsilon)$ (see
  \cref{def:muNelson}), we have that
  \begin{align*}
    \meanlr{\psi}{W_{\varepsilon_n}}{\varphi}
    &= \meanlr{\psi(x)\otimes\Psi_{\varepsilon_n}}{ \lf( a^{\dagger}_{\varepsilon_n}(\vec{w}_x) \ri)^2 + \lf(a_{\varepsilon_n}(\vec{w}_x) \ri)^2 + 2 a^{\dagger}_{\varepsilon_n}(\vec{w}_x)a_{\varepsilon_n}(\vec{w}_x)}{ \varphi(x)\otimes\Psi_{\varepsilon_n}} \\
    &= \braketr{\Psi_{\varepsilon_n}}{\meanlr{\psi(x)}{\lf( a^{\dagger}_{\varepsilon_n}(\vec{w}_x) \ri)^2 + \lf( a_\varepsilon(\vec{w}_x) \ri)^2 + 2 a^{\dagger}_{\varepsilon_n}(\vec{w}_x)a_{\varepsilon_n}(\vec{w}_x) }{\varphi(x)}_{L_{x}^{2}}\Psi_{\varepsilon_n}}_{\mathscr{F}_{\eps_n}}\\
    &\xrightarrow[n \to +\infty]{} \int_{\mathfrak{h}_{\omega}} \meanlr{\psi(x)}{ \lf( 2\mathrm{Re}\, \braket{z}{w_{x}} \ri)^{2}}{\varphi(x)} \, \D \mu(z) =\langle \psi \mid W_{\mu}  \varphi \rangle ,
  \end{align*}
  since
  $\overline{\mathcal{F}}(\bar{\psi}\varphi)(k+k')\frac{\chi(k)}{\omega(k)}\frac{\chi(k')}{\omega(k')}\in
  L_{k,k'}^{2}$ and $\mu$ is supported on functions $z$ such that
  $\sqrt{\omega}z\in L^{2}$.


  The weak convergence results follow similarly and the uniform boundedness
  results follow from \cref{prop:estimate-Aepsilon,prop:estimate-Wepsilon,prop:est Beps,prop:unif2-1}.
\end{proof}

\begin{rem}
	\mbox{}	\\
  The previous proof also implies that
  \begin{itemize}
  \item
    $\mathcal{F}\lf(\vec{A}_{\varepsilon_n} \ri) \xrightarrow[n \to +\infty]{} \mathcal{F}\lf(\vec{A}_{\mu} \ri)$
    weakly in~$L^{6/5,2}$,
  \item $ \mathcal{F}\lf({W}_{\varepsilon_n}\ri) \xrightarrow[n \to +\infty]{} \mathcal{F}\lf({W}_{\mu}\ri) $
    weakly in~$L^{3/2,2}$,
    \item  $\vec{B}_{\varepsilon_n}  \xrightarrow[n \to +\infty]{} \vec{B}_\mu$ weakly in $L^2$.
  \end{itemize}
\end{rem}

We are now ready to prove our main result about the Pauli-Fierz model.

\begin{proof}[Proof of \cref{thm:Quasi-classical-limit-Pauli-Fierz}]

We proceed as in the proof of \cref{thm:Quasiclassical-limit-Nelson}.  Let us fix $\mu\in \mathscr{M}_{\omega}(\Psi_{\varepsilon})$, and the
  sequence $\{\varepsilon_n\}_{n\in \mathbb{N}}$ such that
  $\Psi_{\varepsilon_n} \to \mu$.

  To justify that 
  $H^{\mathrm{PF}}_{\varepsilon_n}$ and
  $H^{\mathrm{PF}}_{\mu}$ define symmetric closed quadratic forms with
  form
  domains $\mathcal{Q}:=H^{1}(\R^3)
  \cap L^{2}(\R^3, U_{+}\, \D x)$, it suffices to argue exactly as in the proof of  \cref{thm:Quasiclassical-limit-Nelson}, using that $H^{\mathrm{PF}}_{\varepsilon_n}-(-\Delta+U_+)$ and $H^{\mathrm{PF}}_{\varepsilon_n}-(-\Delta+U_+)$ are relatively form bounded w.r.t. $-\Delta$ with a relative bound $<1$, which follows from \cref{hyp:A-U} and \cref{prop:estimate-Aepsilon,prop:estimate-Wepsilon,prop:est Beps,prop:unif2-1}.

 Next, we prove $\Gamma$-convergence of
  the quadratic forms $ \meanlr{\varphi}{H^{\mathrm{PF}}_{\varepsilon_n}}{\varphi} $ to $ \meanlr{\varphi}{H^{\mathrm{PF}}_{\mu}}{\varphi}$.

  For the $\Gamma\text{-}\limsup$ we take again the constant sequence and use
  \cref{prop:Convergence-V-1}: for any $\varphi\in \mathcal{Q}$,
  \bdm
     \meanlr{\varphi}{H_{\eps_n}^{\mathrm{PF}}-H^{\mathrm{PF}}_{\mu}}{\varphi} = 2\mathrm{Re} \braket{\sigmav\cdot\vec{P}\varphi}{\sigmav\cdot \lf(\vec{A}_{\mu}-\vec{A}_{\varepsilon_n} \ri)\varphi} +  \meanlr{\varphi}{W_{\varepsilon_n}-W_{\mu}}{\varphi}
    \xrightarrow[n\to +\infty]{}0\;.
  \edm

For the $\Gamma\text{-}\liminf$, we apply \cref{prop:2}: the fact that the assumptions of \cref{prop:2} are satisfied follows from \cref{prop:estimate-Aepsilon,prop:estimate-Wepsilon,prop:est Beps,prop:unif2-1,prop:Convergence-V-1} (in particular, \eqref{eq:unif-bounds} holds with $\delta=3/8$). This concludes the proof.

\end{proof}

\medskip

\begin{footnotesize}
\noindent
\textbf{Acknowledgments.} M.C. \& M.F. acknowledges the supports of PNRR Italia Domani and Next Generation Eu through the ICSC National Research Centre for High Performance Computing, Big Data and Quantum Computing, and of the MUR grant ``Dipartimento di Eccellenza 2023-2027'' of Dipartimento di Matematica, Politecnico di Milano.
\end{footnotesize}

\appendix

\section{Weak $\Gamma$-convergence for Schr\"{o}dinger and Pauli operators}
\label{sec:weak-gamma-conv}\label{sec:Gamma-conv}

In this short appendix, we will prove some results concerning the limes
inferior of a sequence of operators in the weak topology that are useful for
proving $\Gamma$-convergence of quasi-classical operators in the main text. 
We consider
Schr\"{o}dinger or Pauli type operators, with perturbations that are
uniformly KLMN-relatively-small and converge weakly.

We first prove a preliminary useful lemma.
\begin{lem}
  \label{lemma:1}
  Let $Q:\mathcal{Q}\to \mathbb{R}_+$ be a non-negative, densely defined quadratic form, and
  $Q[\,\cdot\, ,\,\cdot\, ]$ the associated sesquilinear form. Then,
  \begin{equation*}
    Q[\psi]=\sup_{\phi\in \mathcal{Q}} \mathrm{Re} \, Q[\phi,2\psi-\phi]\;.
  \end{equation*}
\end{lem}
\begin{proof}
  By the polarization identity, we can write
  \begin{multline*}
    \mathrm{Re} \, Q[\phi,2\psi-\phi] = \tfrac{1}{4}\mathrm{Re} \, \Bigl\{Q[2\psi-\phi+\phi]-Q[2\psi-\phi-\phi]+iQ[2\psi-\phi-i\phi]-iQ[2\psi-\phi-i\phi]\Bigr\}\\=\tfrac{1}{4}\Bigl(Q[2\psi]-Q[2(\psi-\phi)]\Bigr)= Q[\psi]-Q[\psi-\phi]\;.
  \end{multline*}
  Therefore,
  \begin{equation*}
    \sup_{\phi\in \mathcal{Q}} \mathrm{Re} \, Q[\phi,2\psi-\phi]= Q[\psi]-\inf_{\phi\in \mathcal{Q}} Q[\psi-\phi]= Q[\psi]\;,
  \end{equation*}
  since $Q[\,\cdot \,]\geq 0$ (choose $\phi=\psi$).
\end{proof}
Let $\mathfrak{h} = L^2(\mathbb{R}^d;\mathbb{C}^\nu)$ with $\nu=1$ in the case of Schr\"odinger operators and $\nu=2$ for Pauli operators, and
consider 
\begin{equation}
  \label{eq:4}
  H_n= -\Delta+U + V_n\;,
\end{equation}
where $U$ satisfies \cref{hyp:A-U}, $V _n=W_n$ for Schr\"odinger, $V_n=W_n+\nabla\cdot \vec{A}_n+\sigma\cdot \vec{B}_n$ for Pauli, with $W_n$ measurable from $\mathbb{R}^d$ to $\mathbb{R}$, $\nabla\cdot \vec{A}_n=\vec{A}_n\cdot\nabla$, and $\lf( \vec{A}_{n} \ri)_j$, $ \lf( \vec{B}_{n} \ri)_j$, for $1\le j\le d$, measurable from $\mathbb{R}^d$ to $\mathbb{R}$. Let us suppose that $V_n$ converges to
some $V_{\infty}=W_\infty+\nabla\cdot \vec{A}_\infty+\sigma\cdot \vec{B}_\infty$ weakly on $\mathcal{Q}(-\Delta+U)$: for all $\psi,\phi\in \mathcal{Q} (-\Delta+U)$,
\begin{equation*}
  \lim_{n\to \infty} \meanlr{\phi}{V_n}{\psi}_{}= \meanlr{\phi}{V_{\infty}}{\psi}_{}\;.
\end{equation*}
We set $H_\infty:=-\Delta+U+V_\infty$. Furthermore, we suppose that for some $0<\delta<\frac12$ and $\lambda_0>0$,
\begin{equation}\label{eq:unif-bounds}
\begin{split}
&  \Bigl\lVert (-\Delta+\lambda_0)^{-\delta} W_n (-\Delta+\lambda_0)^{-\delta}  \Bigr\rVert_{}^{}\leq C\;, \\
&\Bigl\lVert \vec{A}_n (-\Delta+\lambda_0)^{-\delta}  \Bigr\rVert_{}^{}\leq C\;,\; \\
& \Bigl\lVert (-\Delta+\lambda_0)^{-\delta} \vec{B}_n (-\Delta+\lambda_0)^{-\delta}  \Bigr\rVert_{}^{}\leq C \;,
\end{split}
\end{equation}
all uniformly with respect to $n\in\mathbb{N}\cup\{+\infty\}$. Let us remark that this implies that $V_n$ (for $n\in\mathbb{N}\cup\{+\infty\}$) is a $-\Delta$-relatively small
perturbation  in the sense of quadratic forms: there exist $a<1$ and $b\geq 0$ such
that for all $\psi\in H^1(\mathbb{R}^d)$ and $n\in \mathbb{N}\cup\{+\infty\}$,
\begin{equation}\label{eq:Vn_inf_small}
  \lf| \meanlr{\psi}{V_n}{\psi} \ri| \leq a \meanlr{\psi}{-\Delta}{\psi}+ b\lVert \psi  \rVert_{}^2\;.
\end{equation}
In turn, this implies that there exists $m\in \mathbb{R}$ bounding from
below the spectrum of all $H_n$ and $H_{\infty}$, and that any non-uniformly
bounded sequence in $\mathcal{Q}(-\Delta+U)$ makes $ \meanlr{\psi_n}{H_n}{\psi_n}
 $ diverge as $n\to \infty$. 


\begin{prop}[weak $\Gamma$-lower bound]
  \label{prop:2}
  \mbox{}	\\
  Let $H_n=-\Delta+U+V_n$, and $H_{\infty}=-\Delta+U+V_{\infty}$ be defined as above, with $U$ satisfying \cref{hyp:A-U} and such that:
  $V_n$ converges weakly on $\mathcal{Q}:=\mathcal{Q}(-\Delta+U)$ to $V_{\infty}$, and \eqref{eq:unif-bounds} holds uniformy in $\mathbb{N}$ for
  some $\delta>0$ and $\lambda_0>0$.
Then, for any $\{\psi_n\}_{n \in \N} \subset \mathcal{Q}$ such that
  $\psi_n\xrightarrow[n\to \infty]{\mathrm{w}-L^2}\psi$ and $\psi\in \mathcal{Q}$, it holds that for
  all $\lambda>\max(-m,\lambda_0)$,
  \begin{equation}
    \label{eq:5}
    \liminf_{n\to \infty} \meanlr{\psi_n}{H_n+\lambda}{\psi_n} \geq  \meanlr{\psi}{H_{\infty}+\lambda}{\psi}\;.
  \end{equation} 
\end{prop}
\begin{proof}
  First of all, since a non-uniformly-bounded sequence $\psi_n$ on $\mathcal{Q}:=\mathcal{Q}(-\Delta+U)$
   makes the l.h.s. of the inequality to prove diverge, we can
  restrict to uniformly bounded weakly convergent sequences, hence to
  sequences $\psi_n\xrightarrow[n\to+ \infty]{\mathrm{w}-\mathcal{Q}}\psi$. Let $\psi_{\kappa}\in C_0^{\infty}(\mathbb{R}^d)$ be such that
  \begin{equation*}
    \lVert \psi_{\kappa}-\psi  \rVert_{\mathcal{Q}}^{}< \kappa\;,
  \end{equation*}
  and $\mathrm{supp} (\psi_{\kappa})\subset B_0(R_{\kappa})$, with $R_{\kappa}\to \infty$ as $\kappa\to 0$. 
  
  Using
  \cref{lemma:1}, we can write
  \begin{multline*}
    \meanlr{\psi_n}{H_n+\lambda}{\psi_n} =\sup_{\phi\in \mathcal{Q}} \mathrm{Re} \meanlr{\phi}{H_n+\lambda}{2\psi_n-\phi} \geq \mathrm{Re} \meanlr{\psi_{\kappa}}{H_n+\lambda}{2\psi_n-\psi_{\kappa}} \\=-  \meanlr{\psi_{\kappa}}{H_n+\lambda}{\psi_{\kappa}}+ 2\mathrm{Re} \meanlr{\psi_{\kappa}}{H_n+\lambda}{\psi_n}\;. 
  \end{multline*}
  Now, the first term on the rightmost hand side converges by weak
  convergence of $V_n$:
  \begin{equation*}
    \lim_{n\to \infty} \meanlr{\psi_{\kappa}}{H_n+\lambda}{\psi_{\kappa}} = \meanlr{\psi_{\kappa}}{H_\infty+\lambda}{\psi_{\kappa}}\;.
  \end{equation*}
  We rewrite the remaining term as
  \begin{equation*}
    \mathrm{Re} \meanlr{\psi_{\kappa}}{H_n+\lambda}{\psi_n} = \mathrm{Re} \meanlr{\psi_{\kappa}}{-\Delta+U+\lambda}{\psi_n}+ \mathrm{Re}  \meanlr{\psi_{\kappa}}{V_n}{\psi_n}\;.
  \end{equation*}
  Since $\psi_n\xrightarrow[n\to +\infty]{\mathrm{w}-\mathcal{Q}}\psi$,
  \begin{equation*}
    \mathrm{Re}  \meanlr{\psi_{\kappa}}{-\Delta+U+\lambda}{\psi_n} \xrightarrow[n\to +\infty]{} \mathrm{Re} \meanlr{ \psi_{\kappa}}{-\Delta+U+\lambda}{\psi}\;.
  \end{equation*}
  Finally, it remains to consider $\mathrm{Re} \meanlr{\psi_{\kappa}}{V_n}{\psi_n} = \mathrm{Re} \meanlr{\psi_{\kappa}}{ V_n}{\psi} + \mathrm{Re} \meanlr{\psi_{\kappa}}{V_n}{\psi_n- \psi} $: since $V_n$ converges weakly on $\mathcal{Q}$ to $V_{\infty} $, we have
  \begin{equation*}
    \lim_{n\to +\infty}\mathrm{Re} \meanlr{\psi_{\kappa}}{V_n}{\psi} = \mathrm{Re} \meanlr{\psi_{\kappa}}{V_\infty}{\psi} \;.
  \end{equation*}
Next, we write
  \begin{multline*}
    \mathrm{Re} \meanlr{\psi_{\kappa}}{V_n}{\psi_n- \psi} = \mathrm{Re} \meanlr{\chi_{B_0(R_{\kappa})} \psi_{\kappa}}{V_n}{\psi_n- \psi} \\
    = \mathrm{Re} \meanlr{\psi_{\kappa}}{V_n}{\chi_{B_0(R_{\kappa})}(\psi_n- \psi)} + \mathrm{Re} \braketr{ \psi_{\kappa}}{[\chi_{B_0(R_{\kappa})},V_n] (\psi_n- \psi)} \;,
  \end{multline*}
  where $\chi_{B_0(R_{\kappa})}$ is a smooth characteristic function that is $ \equiv 1 $
  inside $B_0(R_{\kappa})$, and is supported on $B_0(2R_{\kappa})$. 
  For the first term on the r.h.s., observing that $(-\Delta+\lambda_0)^{-1/2}V_n(-\Delta+\lambda_0)^{-\delta}$ is uniformly bounded in $n$ by \eqref{eq:unif-bounds} and that $\|\psi_\kappa\|_{\mathcal{Q}}$ is uniformly bounded in $0<\kappa<1$, we can estimate
  \begin{multline*}
  \lf| \meanlr{\psi_{\kappa}}{V_n}{\chi_{B_0(R_{\kappa})}(\psi_n- \psi)} \ri| \\\le \lf\|\psi_\kappa \ri\|_{\mathcal{Q}} \lf\| (-\Delta+\lambda_0)^{-1/2}V_n(-\Delta+\lambda_0)^{-\delta} \ri\| \lf\| (-\Delta+\lambda_0)^{\delta} \chi_{B_0(R_{\kappa})}(\psi_n- \psi) \ri\| \\
  \le C \lf\| (-\Delta+\lambda_0)^{\delta}\chi_{B_0(R_{\kappa})}(\psi_n- \psi) \ri\|.
  \end{multline*}
  Now we can write
  \begin{align*}
  (-\Delta+\lambda_0)^{\delta}\chi_{B_0(R_{\kappa})}(\psi_n- \psi) 
&  =(-\Delta+\lambda_0)^{-1/2+\delta}\chi_{B_0(R_{\kappa})} (-\Delta+\lambda_0)^{1/2} (\psi_n- \psi) \\
 &\quad - (-\Delta+\lambda_0)^{\delta} \lf[ \chi_{B_0(R_{\kappa})} , (-\Delta+\lambda_0)^{-1/2} \ri] (-\Delta+\lambda_0)^{1/2} (\psi_n- \psi) .
  \end{align*}
  For any fixed $0<\kappa<1$, the first term goes to $0$ in norm as $n\to+\infty$ since $(-\Delta+\lambda_0)^{-1/2+\delta}\chi_{B_0(R_{\kappa})}$ is compact and $(-\Delta+\lambda_0)^{1/2} (\psi_n- \psi)\to 0$ as $n\to+\infty$ weakly in $L^2$. As for the second term, we use the representation formula 
  	\bdm
  		(-\Delta+\lambda_0)^{-1/2}=\pi^{-1}\int_0^\infty s^{-1/2}\big( -\Delta+\lambda_0+s \big )^{-1} \mathrm{d}s,
	\edm 
which gives
  	\bdm
  		\lf[\chi_{B_0(R_{\kappa})} , (-\Delta+\lambda_0)^{-1/2} \ri]=\pi^{-1}\int_0^\infty s^{-1/2}\big( -\Delta+\lambda_0+s \big )^{-1}\lf [-\Delta,\chi_{B_0(R_{\kappa})}\ri]\big( -\Delta+\lambda_0+s \big )^{-1} \mathrm{d}s.
	\edm 
Since $\delta<\frac12$, standard estimates exploiting the scaling properties of the Laplacian resolvent then show that 
    \begin{equation*}
	  \lf\|(-\Delta+\lambda_0)^{\delta} \lf[\chi_{B_0(R_{\kappa})} , (-\Delta+\lambda_0)^{-1/2} \ri] \ri\|\le CR_\kappa^{-1},
  \end{equation*}
  and hence, since in addition $ \lf\|\psi_n \ri\|_{H^1}$ is uniformly bounded in $n$,
  \begin{equation*}
   \lf\|(-\Delta+\lambda_0)^{\delta} \lf[\chi_{B_0(R_{\kappa})} , (-\Delta+\lambda_0)^{-1/2} \ri] (-\Delta+\lambda_0)^{1/2} (\psi_n- \psi) \ri\| \le CR_\kappa^{-1} .
  \end{equation*}
The previous estimates yield
  \begin{equation*}
    \limsup_{n\to +\infty} \lf| \meanlr{\psi_{\kappa} }{V_n\chi_{B_0(R_{\kappa})}}{\psi_n- \psi} \ri| \le CR_\kappa^{-1}.
  \end{equation*}
  
  It remains to consider the term $\mathrm{Re} \braketr{\psi_{\kappa}}{[\chi_{B_0(R_{\kappa})},V_n] (\psi_n- \psi)}$. We compute
  \begin{equation*}
 \mathrm{Re} \braketr{\psi_{\kappa}}{\lf[V_n,\chi_{B_0(R_{\kappa})} \ri] (\psi_n- \psi)} \\= \mathrm{Re} \braketr{ \psi_{\kappa}}{\vec{A}_n\cdot (\nabla\chi_{B_0(R_{\kappa})})(\psi_n- \psi)}\;.
  \end{equation*}
  This term converges to zero by the Rellich-Kondrachov theorem, since $\vec{A}_n\psi_\kappa$ is uniformly bounded in $n$ by \eqref{eq:unif-bounds}, 
  $\nabla\chi_{B_0(R_{\kappa})}$ is a compactly supported function and $\psi_n$ converges weakly to $\psi$ in $H^1$. 
  
  Putting all together, we obtain
  that
  \begin{equation*}
    \liminf_{n\to + \infty} \meanlr{\psi_n}{H_n+\lambda}{\psi_n} \geq - \meanlr{\psi_{\kappa}}{H_{\infty}+\lambda}{\psi_{\kappa}} + 2 \mathrm{Re} \meanlr{\psi_{\kappa}}{H_\infty+\lambda}{\psi} - CR_\kappa^{-1}\;.
  \end{equation*}
  Now, the right hand side converges, as $\kappa\to 0$, to $\meanlr{\psi}{H_{\infty}+\lambda}{\psi} $, thus
  completing the proof.
\end{proof}

{\footnotesize

\begin{thebibliography}{CCFO21}
\providecommand{\url}[1]{\texttt{#1}}
\providecommand{\urlprefix}{URL }
\expandafter\ifx\csname urlstyle\endcsname\relax
  \providecommand{\doi}[1]{doi:\discretionary{}{}{}#1}\else
  \providecommand{\doi}{doi:\discretionary{}{}{}\begingroup
  \urlstyle{rm}\Url}\fi
\providecommand{\eprint}[2][]{\url{#2}}

\bibitem[AN08]{ammari2008ahp}
Z.~Ammari, F.~Nier.
\newblock Mean field limit for bosons and infinite dimensional phase-space
  analysis.
\newblock \emph{Ann. Henri Poincar\'e}
  \href{http://dx.doi.org/10.1007/s00023-008-0393-5}{\textbf{9}~(8), pp.
  1503--1574} (2008).
\newblock \href {https://arxiv.org/abs/0711.4128} {arXiv:0711.4128}.

\bibitem[AN09]{AmmariNier09}
Z.~Ammari, F.~Nier.
\newblock Mean field limit for bosons and propagation of {Wigner} measures.
\newblock \emph{J. Math. Phys.}
  \href{http://dx.doi.org/10.1063/1.3115046}{\textbf{50}~(4), pp. 042,107, 16}
  (2009).

\bibitem[AN11]{AmmariNier11}
Z.~Ammari, F.~Nier.
\newblock Mean field propagation of {Wigner} measures and {BBGKY} hierarchies
  for general bosonic states.
\newblock \emph{J. Math. Pures Appl. (9)}
  \href{http://dx.doi.org/10.1016/j.matpur.2010.12.004}{\textbf{95}~(6), pp.
  585--626} (2011).

\bibitem[AN15]{AmmariNier15}
Z.~Ammari, F.~Nier.
\newblock Mean field propagation of infinite-dimensional {Wigner} measures with
  a singular two-body interaction potential.
\newblock \emph{Ann. Sc. Norm. Super. Pisa, Cl. Sci. (5)}
  \href{http://dx.doi.org/10.2422/2036-2145.201112_004}{\textbf{14}~(1), pp.
  155--220} (2015).

\bibitem[BFP23]{BFP1}
S.~Breteaux, J.~Faupin, J.~Payet.
\newblock Quasi-classical ground states. {I}. {Linearly} coupled
  {Pauli}-{Fierz} {Hamiltonians}.
\newblock \emph{Doc. Math.}
  \href{http://dx.doi.org/10.4171/DM/929}{\textbf{28}~(5), pp. 1191--1233}
  (2023).

\bibitem[BFP24]{BFP2}
S.~Breteaux, J.~Faupin, J.~Payet.
\newblock {Q}uasi-classical {G}round {S}tates. {II}. {S}tandard {M}odel of
  {N}on-relativistic {QED}.
\newblock \emph{Ann. Inst. Fourier}  (2024).
\newblock
  \href{https://aif.centre-mersenne.org/articles/10.5802/aif.3667/}{Online
  first}.

\bibitem[BLNS17]{BezLeeNakamuraSawano17}
N.~{Bez}, S.~{Lee}, S.~{Nakamura}, Y.~{Sawano}.
\newblock Sharpness of the {B}rascamp-{L}ieb inequality in {L}orentz spaces.
\newblock \emph{{Electron. Res. Announc. Math. Sci.}}
  \href{http://dx.doi.org/10.3934/era.2017.24.006}{\textbf{24}, pp. 53--63}
  (2017).

\bibitem[CCFO21]{carlone2021sima}
R.~Carlone, M.~Correggi, M.~Falconi, M.~Olivieri.
\newblock Emergence of {T}ime-{D}ependent {P}oint {I}nteractions in {P}olaron
  {M}odels.
\newblock \emph{SIAM J. Math. Anal.}
  \href{http://dx.doi.org/10.1137/20M1381344}{\textbf{53}~(4), pp. 4657--4691}
  (2021).
\newblock \href {https://arxiv.org/abs/1904.11012} {arXiv:1904.11012}.

\bibitem[CF18]{correggi2017ahp}
M.~Correggi, M.~Falconi.
\newblock {Effective Potentials Generated by Field Interaction in the
  Quasi-Classical Limit}.
\newblock \emph{Ann. Henri Poincar{\'e}}
  \href{http://dx.doi.org/10.1007/s00023-017-0612-z}{\textbf{19}~(1), pp.
  189--235} (2018).
\newblock \href {https://arxiv.org/abs/1701.01317} {arXiv:1701.01317}.

\bibitem[CFO19]{correggi2017arxiv}
M.~Correggi, M.~Falconi, M.~Olivieri.
\newblock {Magnetic Schr\"odinger Operators as the Quasi-Classical Limit of
  Pauli-Fierz-type Models}.
\newblock \emph{J. Spectr. Theory}
  \href{http://dx.doi.org/10.4171/JST/277}{\textbf{9}~(4), pp. 1287--1325}
  (2019).
\newblock \href {https://arxiv.org/abs/1711.07413} {arXiv:1711.07413}.

\bibitem[CFO23a]{correggi2023apde}
M.~Correggi, M.~Falconi, M.~Olivieri.
\newblock {G}round {S}tate {P}roperties in the {Q}uasi-{C}lassical {R}egime.
\newblock \emph{Anal. PDE}
  \href{http://dx.doi.org/10.2140/apde.2023.16.1745}{\textbf{16}~(8), pp.
  1745--1798} (2023).
\newblock \href {https://arxiv.org/abs/2007.09442} {arXiv:2007.09442}.

\bibitem[CFO23b]{correggi2023jems}
M.~Correggi, M.~Falconi, M.~Olivieri.
\newblock {Quasi-Classical Dynamics}.
\newblock \emph{J. Eur. Math. Soc. (JEMS)}
  \href{http://dx.doi.org/10.4171/JEMS/1197}{\textbf{25}~(2), pp. 731--783}
  (2023).
\newblock \href {https://arxiv.org/abs/1909.13313} {arXiv:1909.13313}.

\bibitem[CFO24]{correggi2024inprep}
M.~Correggi, M.~Falconi, M.~Olivieri.
\newblock Quasi-classics of the renormalized nelson model.
\newblock \emph{In preparation}  (2024).

\bibitem[DM93]{dalmaso1993pnde}
G.~Dal~Maso.
\newblock \emph{An introduction to {$\Gamma$}-convergence}, \emph{Progress in
  Nonlinear Differential Equations and their Applications}, vol.~8
  (Birkh\"{a}user Boston, Inc., Boston, MA, 1993).
\newblock ISBN 0-8176-3679-X, xiv+340 pp.

\bibitem[Fal18a]{falconi2017ccm}
M.~Falconi.
\newblock {Concentration of cylindrical Wigner measures}.
\newblock \emph{Commun. Contemp. Math.}
  \href{http://dx.doi.org/10.1142/S0219199717500559}{\textbf{20}~(5), p.
  1750,055} (2018).
\newblock \href {https://arxiv.org/abs/1704.07676} {arXiv:1704.07676}.

\bibitem[Fal18b]{falconi2017arxiv}
M.~Falconi.
\newblock {Cylindrical Wigner measures}.
\newblock \emph{Doc. Math.}
  \href{http://dx.doi.org/10.25537/dm.2018v23.1677-1756}{\textbf{23}, pp.
  1677--1756} (2018).
\newblock \href {https://arxiv.org/abs/1605.04778} {arXiv:1605.04778}.

\bibitem[Gra08]{Grafakos08}
L.~Grafakos.
\newblock \emph{Classical {Fourier} analysis}, \emph{Grad. Texts Math.}, vol.
  249 (New York, NY: Springer, 2008), 2nd ed. edn.
\newblock ISBN 978-0-387-09431-1.

\bibitem[LR02]{Lemarie-Rieusset02}
P.~G. Lemari{\'e}-Rieusset.
\newblock \emph{Recent developments in the {Navier}-{Stokes} problem},
  \emph{Chapman Hall/CRC Res. Notes Math.}, vol. 431 (Boca Raton, FL: Chapman
  \& Hall/CRC, 2002).
\newblock ISBN 1-58488-220-4.

\bibitem[Nel64]{nelson1964jmp}
E.~Nelson.
\newblock Interaction of nonrelativistic particles with a quantized scalar
  field.
\newblock \emph{J. Math. Phys.} \textbf{5}, pp. 1190--1197 (1964).

\bibitem[O'N63]{ONeil63}
R.~O'Neil.
\newblock Convolution operators and {L}(p, q) spaces.
\newblock \emph{Duke Math. J.}
  \href{http://dx.doi.org/10.1215/S0012-7094-63-03015-1}{\textbf{30}, pp.
  129--142} (1963).

\bibitem[PF38]{pauli1938nc}
W.~Pauli, M.~Fierz.
\newblock {Zur Theorie der Emission langwelliger Lichtquanten}.
\newblock \emph{Il Nuovo Cimento} \textbf{15}~(3), pp. 167--188 (1938).

\bibitem[RS75]{reed1975II}
M.~Reed, B.~Simon.
\newblock \emph{Methods of modern mathematical physics. {II}. {F}ourier
  analysis, self-adjointness} (Academic Press, New York, 1975), xv+361 pp.

\bibitem[Spo04]{spohn2004dcp}
H.~Spohn.
\newblock \emph{Dynamics of charged particles and their radiation field}
  (Cambridge University Press, Cambridge, 2004).
\newblock ISBN 0-521-83697-2, xvi+360 pp.

\bibitem[Yap69]{Yap69}
L.~Y.~H. Yap.
\newblock Some remarks on convolution operators and {L}(p,q) spaces.
\newblock \emph{Duke Math. J.}
  \href{http://dx.doi.org/10.1215/S0012-7094-69-03677-1}{\textbf{36}, pp.
  647--658} (1969).

\end{thebibliography}

}

\end{document}